\documentclass{amsart}[12 pt]

\usepackage{latexsym}

\usepackage{amsmath,amssymb,dsfont,amsfonts}
\usepackage{mathrsfs}
\usepackage{verbatim}
\usepackage{graphicx}
\usepackage{graphics}
\usepackage{geometry}
\usepackage{booktabs}
\usepackage{enumitem}

\usepackage[colorlinks]{hyperref}

\newtheorem{proposition}{Proposition}[section]
\newtheorem{theorem}{Theorem}[section]
\newtheorem{definition}{Definition}[section]
\newtheorem{corollary}{Corollary}[section]
\newtheorem{lemma}{Lemma}[section]
\newtheorem{remark}{Remark}[section]
\newtheorem{example}{Example}[section]

\DeclareMathOperator*{\esssup}{ess\,sup} 

\numberwithin{equation}{section}

\allowdisplaybreaks

\begin{document}
\markboth{}{}

\title[Discrete-type approximations for non-Markovian optimal stopping problems: Part II] {Discrete-type approximations for non-Markovian optimal stopping problems: Part II}


\author{S\'ergio C. Bezerra}

\address{Departamento de computa\c{c}\~ao cient\'ifica, Universidade Federal da Para\'iba, Rua dos Escoteiros, Jo\~ao Pessoa - Para\'iba, Brazil}\email{sergio@ci.ufpb.br}

\author{Alberto Ohashi}
\address{Departamento de Matem\'atica, Universidade de Bras\'ilia, 70910-900, Bras\'ilia - Distrito Federal, Brazil}\email{amfohashi@gmail.com}

\author{Francesco Russo}
\address{ENSTA ParisTech, Unit\'e de Math\'ematiques appliqu\'ees,
828, Boulevard des Mar\'echaux, F-91120, Palaiseau, France.}\email{francesco.russo@ensta-paristech.fr}

\author{Francys de Souza}

\address{Instituto de Matem\'atica, Estat\'istica e Computa\c{c}\~ao Cient\'ifica. Universidade de Campinas, 13083-859, Campinas - SP, Brazil}\email{francysouz@gmail.com}

\thanks{}
\date{\today}


\begin{center}
\end{center}

\begin{abstract}
In this paper, we present a Longstaff-Schwartz-type algorithm for optimal stopping time problems based on the Brownian motion filtration. The algorithm is based on Le\~ao, Ohashi and Russo \cite{LEAO_OHASHI2017.3} and, in contrast to previous works, our methodology applies to optimal stopping problems for fully non-Markovian and non-semimartingale state processes such as functionals of path-dependent stochastic differential equations and fractional Brownian motions. Based on statistical learning theory techniques, we provide overall error estimates in terms of concrete approximation architecture spaces with finite Vapnik-Chervonenkis dimension. Analytical properties of continuation values for path-dependent SDEs and concrete linear architecture approximating spaces are also discussed.
\end{abstract}

\maketitle

\section{Introduction}
Optimal stopping is a quite popular type of stochastic control problem with many applications in applied sciences. A general optimal stopping problem can be formulated as follows. Let $(\Omega,\mathcal{F},\mathbb{P})$ be a complete probability space and let $\mathbb{F} = (\mathcal{F}_t)_{t\ge 0}$ be the natural augmented filtration generated by a $d$-dimensional standard Brownian motion $B$ and let $Z$ be an $\mathbb{F}$-adapted process. For a given $T>0$, one has to find

\begin{equation}\label{intr1}
\sup_{\tau\in \mathcal{T}_0(\mathbb{F})}\mathbb{E}[Z(\tau)]
\end{equation}
where $\mathcal{T}_0(\mathbb{F})$ denotes the set of all $\mathbb{F}$-stopping times taking values on the compact set $[0, T]$. Under mild integrability condition on $Z$, it is well-known the Snell envelope process of $Z$, denoted by $S$, is the minimal supermartingale of class (D) which dominates $Z$. Moreover, $S$ fully characterizes the optimal stopping problem in this very general setting. See e.g Karatzas and Shreve \cite{karatzas} and Lamberton \cite{lamberton} for further details.

A successful construction of the process $S$ leads to the solution of the optimal stopping problem. For instance, in case $Z  = g(X)$ for a function $g$ and a Markov process $X$, $S$ is then characterised by the least excessive (superharmonic) function $V (\cdot)$ that majorizes $g(\cdot)$ (see e.g Peskir and Shiryaev \cite{peskir}). In this Markovian context, PDE methods come into play in order to obtain value functions associated with optimal stopping problems, specially in low dimensions. In higher dimensions, a popular approach is to design Monte Carlo schemes for optimal stopping problems written on a Markovian state $X$. The literature on this research topic is vast. For an overview of the literature, we refer e.g to Bouchard and Warin \cite{bouchard}, Kholer \cite{kohler} and other references therein.

Numerical methods for optimal stopping have been widely studied by many authors and in different contexts. For instance, the least squares regression method for pricing American options has been widely studied in the Finance literature. The origins of the method can be found in the works of Carriere \cite{carriere}, Tsitsiklis and Van Roy \cite{tsi}, Longstaff and Schwartz \cite{longstaff} and Cl\'ement, Lamberton and Protter \cite{clement}. Basically, the method seeks a way of computing conditional expectations (the so-called continuation values) needed in the valuation process either directly as in \cite{longstaff, clement}, or indirectly through the value function as in \cite{tsi}. Egloff \cite{egloff} achieves a major advancement by introducing the dynamic look-ahead algorithm which includes both the Tsitsiklis-Van Roy and Longstaff-Schwartz algorithms as special cases. The key assumptions in Egloff \cite{egloff} require the architecture approximating sets designed to approach the continuation values to be closed, convex and uniformly bounded with a finite Vapnik-Chervonenkis dimension (VC-dimension). Later, Zanger \cite{zanger,zanger1,zanger2} presents more general results by employing nonlinear approximating architecture spaces and not necessarily convex and closed. Sequential design schemes are introduced by Gramacy and Ludkovski \cite{gramacy} in order to infer the underlying optimal stopping boundary. See also Hu and Ludkovski \cite{hu} for ranking generic response surfaces. We also draw attention to numerical methods based on the so-called dual approach (see Rogers \cite{rogers}) of the optimal stopping problem. In this direction, see e.g Belomestny \cite{belo} and Belomestny, Schoenmakers and Dickmann \cite{belo1} and other references therein. The common assumption in all these works is the Markov property on the underlying state reward process which allows us to handle continuation values in the dynamic programming algorithm.

In this work, we present a Monte Carlo scheme specially designed to solve optimal stopping problems of the form (\ref{intr1}) where the reward process
\begin{equation}\label{intr2}
Z = F(X)
\end{equation}
is a path-dependent functional of a general continuous process $X$ which is adapted to the Brownian filtration $\mathbb{F}$. The main contribution of this article is the development of a feasible Longstaff-Schwartz-type algorithm for fully non-Markovian states $X$. We are particularly interested in the case when $X$ cannot be reduced to vectors of Markov processes and $F$ may depend on the whole path of $X$. Monte Carlo schemes for optimal stopping problems driven by Markovian states have been intensively studied over the last two decades as described above, but to the best of our knowledge, a concrete and thorough analysis in the path dependent case has not yet been made. In particular, there is a gaping lack of results for truly non-Markovian systems, where the state $X$ cannot be transformed to a Markov process either because it will end up at an infinite-dimensional dynamics or due to the lack of observability. The present article is an attempt to close this gap, at least for the particular case when the underlying filtration is generated by the Brownian motion, i.e., we are restricted to continuous state processes without the presence of jumps.

Optimal stopping problems (\ref{intr1}) written on non-Markovian reward processes of the form (\ref{intr2}) arise in many contexts and applications. One important case appears in Finance. For instance, for American-style options, the price is the supremum over a large range of possible stopping times of the discounted expected payoff under a risk-neutral measure. In order to compute (\ref{intr1}), one has to invoke dynamic programming principle and a feasible numerical scheme for conditional expectations plays a key role in pricing American options. In the classical Markovian case, plenty of methods are available (see e.g Bouchard and Warin \cite{bouchard}). Under the presence of stochastic volatility or volatility structure depending on the whole asset price path, the asset price process becomes non-Markovian which makes the numerical analysis much harder. One typical way to overcome the lack of Markov property in stochastic volatility models is to assume the volatility variable is fully observable or at least it can be consistently estimated from implied volatility surfaces. More sophisticated methods based on nonlinear filtering and hidden Markov processes techniques can also be employed. In this direction, see e.g Rambharat and Brockwell \cite{ramb}, Ludkovski \cite{ludo}, Song, Liang and Liu \cite{song}, Ye and Zhou \cite{ye} and other references therein. In more complex cases, even if one assumes an observable volatility structure, one cannot reduce the problem (\ref{intr1}) to a Markovian case without adding infinitely many degrees of freedom. This type of phenomena arises when the volatility is a functional of the fractional Brownian motion $B_H$ as described e.g in Bayer, Friz and Gatheral \cite{bayer}, Gatheral \cite{gatheral}, Chronopoulou and Viens \cite{viens}, Forde and Zhang \cite{forde} and other references therein. Moreover, since fractional Brownian motion is neither a semimartingale nor a Markov process for $H\neq 1/2$, a concrete development of a Monte Carlo method is a highly non-trivial task.



In this work, we present a Monte Carlo scheme which applies to quite general states including path-dependent payoff functionals of path-dependent stochastic differential equations (henceforth abbreviated by SDEs), stochastic volatility and other non-Markovian systems driven by Brownian motion. The Monte Carlo scheme designed in this work is based on the methodology developed by Le\~ao and Ohashi \cite{LEAO_OHASHI2013, LEAO_OHASHI2017}, Le\~ao, Ohashi and Simas \cite{LEAO_OHASHI2017.1} and Le\~ao, Ohashi and Russo \cite{LEAO_OHASHI2017.3}. In Le\~ao, Ohashi and Russo \cite{LEAO_OHASHI2017.3}, the authors present a discretization method which yields a systematic way to approach fully non-Markovian optimal stopping problems based on the filtration $\mathbb{F}$. The philosophy is to consider the supermartingale Snell envelope

$$S(t) = \esssup_{t\in \mathcal{T}_t(\mathbb{F})}\mathbb{E}\big[Z( \tau)|\mathcal{F}_t\big];0\le t\le T$$
based on a continuous reward process $Z$ viewed as a generic non-anticipative functional of the Brownian  motion $B$, i.e., $Z= Z(B)$. In contrast to the standard literature on regression methods on Markovian optimal stopping, in our approach the relevant structure to be analyzed is the state noise $B$ and the reward process $Z$ is interpreted as a simple functional which ``transports'' $B$ into the system. Therefore, the underlying state space is infinite-dimensional. By using techniques developed by \cite{LEAO_OHASHI2017.1, LEAO_OHASHI2017.3}, we are able to reduce the dimension of the Brownian noise by means of a suitable discrete-time filtration generated by

\begin{equation}\label{AknIntro}
\mathcal{A}^k_n:= \Big(\Delta T^k_1, \eta^k_1, \ldots, \Delta T^k_n, \eta^k_n\Big); 1\le n \le e(k,T)
\end{equation}
where $(T^k_n)_{k,n\ge 1}$ is a suitable family of stopping times with explicit density functions (see e.g Burq and Jones \cite{Burq_Jones2008} and Milstein and Tretyakov \cite{milstein}), $\eta^k_i$ is a discrete random variable which encodes the sign and the jumping coordinate of a suitable approximating martingale at time $i=1, \ldots, e(k,T)$, $e(k,T):=d\lceil T\epsilon_k^{-2}\rceil$ encodes the necessary number of periods to recover the optimal stopping problem (\ref{intr1}) on a given interval $[0,T]$ for a given choice of a sequence $\epsilon_k\downarrow 0$ as $k\rightarrow +\infty$. Theorem \ref{disintegrationTH} presents a closed form expression for transition probabilities for (\ref{AknIntro}) so that the discretization procedure is feasible from a computational point of view (although in a high-dimensional setup). The usual Markov chain state dynamics $X$ onto $\mathbb{R}^n$ is replaced by $\{\mathcal{A}^k_i;1\le i\le e(k,T)\}$ taking values on the high-dimensional state space

$$\mathbb{S}\times \ldots \times \mathbb{S}\quad \big(e(k,T)-\text{fold cartesian product}\big)$$
where $\mathbb{S} :=(0,+\infty)\times \mathbb{I}$ and
$$\mathbb{I}:=\Big\{ (x_1, \ldots, x_d); x_\ell\in \{-1,0,1\}~\forall \ell \in \{1,\ldots, d\}~\text{and}~\sum_{j=1}^d|x_j|=1   \Big\}.$$

With the information set $\{\mathcal{A}^k_i;1\le i\le e(k,T)\}$ at hand, we design what we call an \textit{imbedded discrete structure} (see Definition \ref{GASdef}) for the reward process $Z$. As demonstrated by the works \cite{LEAO_OHASHI2017.1, LEAO_OHASHI2017.2, LEAO_OHASHI2017.3}, this type of structure exists under rather weak regularity assumptions on $Z$ such as path continuity and mild integrability hypotheses. However, we recall the methodology requires effort on the part of the ``user'' in order to specify the best structure that is suitable for analyzing a given state $X$ at hand.

Once a structure is fixed for (\ref{intr2}), Theorem \ref{mainTHLS} shows that a Longstaff-Schwartz-type algorithm associated with the primitive state variables $\mathcal{A}^k$ converges a.s. to the optimal value $S(0)$ as the number of simulated paths $N$ and the discretization level $k$ goes to infinity. Proposition \ref{maincor} and Corollary \ref{contVABS} provide error estimates of a Longstaff-Schwartz-type algorithm for the optimal stopping problem associated with a Snell envelope-type process written on a given imbedded discrete structure w.r.t. $Z$. In order to prove convergence of the Monte Carlo scheme with explicit error estimates, we make use of statistical learning theory techniques originally employed very successfully by Egloff \cite{egloff,egloff1} and Zanger \cite{zanger,zanger1, zanger2} based on Markov chain approximations for Markov diffusions. We show that the same philosophy employed by Zanger \cite{zanger,zanger1,zanger2} can be used to prove convergence of a Longstaff-Schwartz-type algorithm associated with rather general classes of reward processes $Z=F(X)$ where $X$ is an $\mathbb{R}^n$-valued $\mathbb{F}$-adapted continuous process. By means of statistical learning theory techniques, we get precise error estimates (w.r.t. $N$) for each discretization level $k$ encoded by $\epsilon_k$. Similar to the classical Markovian case, the regularity of continuation values related to our backward dynamic programming equation plays a key role on the overall error estimates. Theorem \ref{regPSDE} presents Sobolev-type regularity of the continuation values for a concrete example of a reward $Z$ composed with a non-Markovian path-dependent SDE $X$ driven by the Brownian motion.

It is important to stress that the numerical scheme designed in this paper is to compute the optimal value (\ref{intr1}). In other words, the observer is able to keep track of the control problem (e.g American option) over the \textit{whole} period $[0,T]$ and not only at a given set of exercise discrete times $t_0< \ldots < t_n $ (e.g Bermudan option-style). We point out that there is no conceptual obstruction to treat the discrete-case as well.

This paper is organized as follows. In Section \ref{preliminaries}, we recall some basic objects from Le\~ao, Ohashi and Russo \cite{LEAO_OHASHI2017.3}. In Section \ref{DPPsection}, we define the Longstaff-Schwartz algorithm. In Section \ref{errorESTSECTION}, we present convergence of the method and error estimates. In section \ref{regCVsection}, we present concrete linear architecture approximating spaces and smoothness of continuation values related to a optimal stopping problem arising from a path-dependent SDE. The proof of Theorem \ref{disintegrationTH} is presented in the Appendix section.

\





\section{Preliminaries}\label{preliminaries}
In order to make this paper self-contained, we recall the basic objects employed by Le\~ao, Ohashi and Simas \cite{LEAO_OHASHI2017.1} and Le\~ao, Ohashi and Russo \cite{LEAO_OHASHI2017.3}. For the sake of completeness, we provide here a list of the basic objects that we use in this article. Throughout this work, we are going to fix a $d$-dimensional Brownian motion $B = \{B^{1},\ldots,B^{d}\}$ on the usual stochastic basis $(\Omega, \mathbb{F}, \mathbb{P})$, where $\Omega$ is the space $C([0,+\infty);\mathbb{R}^d) := \{f:[0,+\infty) \rightarrow
\mathbb{R}^d~\text{continuous}\}$, equipped with the usual topology
of uniform convergence on compact intervals, $\mathbb{P}$ is the Wiener measure
on $\Omega$ such that $\mathbb{P}\{B(0) = 0\}=1 $ and $\mathbb{F}:=(\mathcal{F}_t)_{t\ge 0}$ is the
usual $\mathbb{P}$-augmentation of the natural filtration generated by the Brownian motion.

For a fixed positive sequence $\epsilon_k$ such that $\sum_{k\ge 1}\epsilon^2 _k <+\infty$ and for each $j = 1, \ldots, d$, we define $T^{k,j}_0 := 0$ a.s. and we set
\begin{equation}\label{stopping_times}
T^{k,j}_n := \inf\big\{T^{k,j}_{n-1}< t <\infty;  |B^{j}(t) - B^{j}(T^{k,j}_{n-1})| = \epsilon_k\big\}, \quad n \ge 1.
\end{equation}
For each $j\in \{1,\ldots,d \}$, the family $(T^{k,j}_n)_{n\ge 0}$ is a sequence of $\mathbb{F}$-stopping times and the strong Markov property implies that $\{T^{k,j}_n - T^{k,j}_{n-1}; n\ge 1\}$ is an i.i.d. sequence with the same distribution as $T^{k,j}_1$. Moreover, $T^{k,j}_1$ is an absolutely continuous random variable (see Burq and  Jones \cite{Burq_Jones2008}).


From this family of stopping times, we define $A^{k}:=(A^{k,1},\ldots, A^{k,d})$ as the $d$-dimensional step process whose components are given by

\begin{equation}\label{rw}
A^{k,j} (t) := \sum_{n=1}^{\infty}\epsilon_k~\sigma^{k,j}_n1\!\!1_{\{T^{k,j}_n\leq t \}};~t\ge0,
\end{equation}
where

\begin{equation}\label{sigmakn}
\sigma^{k,j}_n:=\left\{
\begin{array}{rl}
1; & \hbox{if} \ B^{j} (T^{k,j}_n) - B^{j} (T^{k,j}_{n-1}) > 0 \\
-1;& \hbox{if} \ B^{j} (T^{k,j}_n) - B^{j} (T^{k,j}_{n-1})< 0, \\
\end{array}
\right.
\end{equation}
for $k,n\ge 1$ and $j=1, \ldots , d$. Let $\mathbb{F}^{k,j} := \{ \mathcal{F}^{k,j}_t; t\ge 0 \} $ be the natural filtration generated by $\{A^{k,j}(t);  t \ge 0\}$. One should notice that $\mathbb{F}^{k,j}$ is a filtration of discrete type~(see Section 4 (Chap 11) and Section 5 (Chap 5) in He, Wang and Yan ~\cite{he}) in the sense that

\[
\mathcal{F}^{k,j}_t = \Big\{\bigcup_{\ell=0}^\infty D_\ell \cap \{T^{k,j}_{\ell} \le t < T^{k,j}_{\ell+1}\}; D_\ell\in \mathcal{F}^{k,j}_{T^{k,j}_\ell}~\text{for}~\ell \ge 0 \Big\},~t\ge 0,
\]
where $\mathcal{F}^{k,j}_0 = \{\Omega, \emptyset \}$ and $\mathcal{F}^{k,j}_{T^{k,j}_m}=\sigma(T^{k,j}_1, \ldots, T^{k,j}_m, \sigma^{k,j}_1, \ldots, \sigma^{k,j}_m)$ for $m\ge 1$ and $j=1,\ldots, d$. From Th 5.56 in He, Wang and Yan \cite{he}, we know that


$$\mathcal{F}^{k,j}_{T^{k,j}_m}\cap\big\{T^{k,j}_m \le t < T^{k,j}_{m+1}\big\} =\mathcal{F}^{k,j}_t\cap \big\{T^{k,j}_m \le t < T^{k,j}_{m+1}\big\},$$
for each $m\ge 0$ and $j=1,\ldots, d$. In this case, $\mathbb{F}^{k,j}$ is a jumping filtration in the sense of Jacod and Shiryaev \cite{jacod}. One can easily check (see Lemma 2.1 in Le\~ao, Ohashi and Simas \cite{LEAO_OHASHI2017.1}) that $A^{k,j}$ is an $\mathbb{F}^{k,j}$-square-integrable martingale over compact sets for every $k\ge 1$ and $1\le j\le d$.

The multi-dimensional filtration generated by $A^k$ is naturally characterized as follows. Let $\mathbb{F}^{k} := \{\mathcal{F}^{k}_t ; 0 \leq t <\infty\}$ be the product filtration given by $\mathcal{F}^{k}_t := \mathcal{F}^{k,1}_t \otimes\mathcal{F}^{k,2}_t\otimes\cdots\otimes\mathcal{F}^{k,d}_t$ for $t\ge 0$. Let $\mathcal{T}:=\{T^{k}_m; m\ge 0\}$ be the order statistics based on the family of random variables $\{T^{k,j}_\ell; \ell\ge 0 ;j=1,\ldots,d\}$. That is, we set $T^{k}_0:=0$,

$$
T^{k}_1:= \inf_{1\le j\le d}\Big\{T^{k,j}_1 \Big\},\quad T^{k}_n:= \inf_{\substack {1\le j\le d\\ m\ge 1} } \Big\{T^{k,j}_m ; T^{k,j}_m \ge  T^{k}_{n-1}\Big\}
$$
for $n\ge 1$. In this case, $\mathcal{T}$ is the partition generated by all stopping times defined in \eqref{stopping_times}. By the independence of the family $\{B^1,\ldots, B^d\}$, the elements of $\mathcal{T}$ are almost surely distinct for every $k\ge 1$.

The structure $\mathscr{D} :=\{\mathcal{T}, A^{k,j}; 1\le j\le d, k\ge 1\}$ is a \textit{discrete-type skeleton} for the Brownian motion in the language of Le\~ao, Ohashi and Simas \cite{LEAO_OHASHI2017.1}. Throughout this article, we set

$$\Delta T^{k}_n: = T^{k}_n - T^{k}_{n-1}, \Delta T^{k,j}_n: = T^{k,j}_n - T^{k,j}_{n-1}; 1\le j\le d, n\ge 1,$$
and $N^{k}(t):=\max\{n; T^{k}_n\le t\}; t\ge 0$. We also denote

$$
\eta^{k,j}_n:=\left\{
\begin{array}{rl}
1; & \hbox{if} \  \Delta A^{k,j} (T^k_n)>0 \\
-1;& \hbox{if} \  \Delta A^{k,j} (T^k_n)< 0 \\
0;& \hbox{if} \ \Delta A^{k,j} (T^k_n)=0,
\end{array}
\right.
$$
and $\eta^k_n:=\big(\eta^{k,1}_n, \dots, \eta^{k,d}_n\big); n,k\ge 1$. Let us define
$$\mathbb{I}_k:=\Big\{ (i^k_1, \ldots, i^k_d); i^k_\ell\in \{-1,0,1\}~\forall \ell \in \{1,\ldots, d\}~\text{and}~\sum_{j=1}^d|i^k_j|=1   \Big\}$$
and $\mathbb{S}_k:=(0,+\infty)\times \mathbb{I}_k$. Let us define $\aleph: \mathbb{I}_k\rightarrow \{1,\ldots, d\}\times\{-1,1\}$ by

\begin{equation}\label{alephamap}
\aleph(\tilde{i}^{k}):=\big(\aleph_1(\tilde{i}^{k}),\aleph_2(\tilde{i}^{k})\big):=(j,r),
\end{equation}
where $j\in\{1,\ldots, d\}$ is the coordinate of $\tilde{i}^k\in \mathbb{I}_k$ which is different from zero and $r\in\{-1,1\}$ is the sign of $\tilde{i}^k$ at the coordinate $j$.

The $n$-fold Cartesian product of $\mathbb{S}_k$ is denoted by $\mathbb{S}_k^n$ and a generic element of $\mathbb{S}^n_k$ will be denoted by $$\mathbf{b}^k_n := (s^k_1,\tilde{i}^k_1, \ldots, s^k_n, \tilde{i}^k_n)\in \mathbb{S}^n_k$$
where $(s^k_r,\tilde{i}^k_r)\in (0,+\infty)\times \mathbb{I}_k$ for $1\le r\le n$. The driving noise in our methodology is given by the following discrete-time process

$$\mathcal{A}^k_n:= \Big(\Delta T^k_1, \eta^k_1, \ldots, \Delta T^k_n, \eta^k_n\Big)\in \mathbb{S}^n_k~a.s.$$
One should notice that $$\mathcal{F}^k_{T^k_n} = (\mathcal{A}^k_n)^{-1}(\mathcal{B}(\mathbb{S}^n_k))$$
where $\mathcal{B}(\mathbb{S}^k_n)$ is the Borel sigma algebra generated by $\mathbb{S}^n_k; n\ge 1$. We set $\mathcal{A}^k_0:=\textbf{0}$ (null vector in~$\mathbb{R}\times\mathbb{R}^d$) and $\mathbb{S}^0_k:=\{\textbf{0}\}$.


The law of the system will evolve according to the following probability measure defined on

$$\mathbb{P}^k_n(E):=\mathbb{P}\circ \mathcal{A}^k_n(E):=\mathbb{P}\{\mathcal{A}^k_n\in E\}; E\in \mathcal{B}(\mathbb{S}^n_k), n\ge 1.$$
The main goal of this article is to treat non-Markovian systems. A priori there is no Markovian semigroup available for us coming from the reward process, so that we need to replace this classical notion in our context. By the very definition,

$$\mathbb{P}^k_{n}(\cdot) = \mathbb{P}^k_{r}(\cdot\times \mathbb{S}^{r-n}_k)$$
for any $r> n\ge 1$.
More importantly, $\mathbb{P}^k_{r}(\mathbb{S}^{n}_k\times \cdot)$ is a regular measure and $\mathcal{B}(\mathbb{S}_k)$ is countably generated, then it is known (see e.g III. 70-73 in Dellacherie and Meyer \cite{dellacherie2}) there exists ($\mathbb{P}^k_{n}$-a.s. unique) a disintegration $\nu^k_{n,r}: \mathcal{B}(\mathbb{S}^{r-n}_k)\times\mathbb{S}^{n}_k\rightarrow[0,1]$ which realizes

$$\mathbb{P}^k_{r}(D) = \int_{\mathbb{S}^{n}_k}\int_{\mathbb{S}^{r-n}_k} 1\!\!1_{D}(\textbf{b}^k_{n},q^k_{n,r})\nu^k_{n,r} (dq^k_{n,r}|\textbf{b}^k_{n})\mathbb{P}^k_{n}(d\textbf{b}^k_{n})$$
for every $D\in \mathcal{B}(\mathbb{S}^{r}_k)$, where $q^k_{n,r}$ is the projection of $\textbf{b}^k_r$ onto the last $(r-n)$ components, i.e., $q^k_{n,r} = (s^k_{n+1},\tilde{i}^k_{n+1}, \ldots,s^k_{r},\tilde{i}^k_{r} )$ for a list $\textbf{b}^k_r = (s^k_1,\tilde{i}^k_1, \ldots, s^k_r,\tilde{i}^k_r)\in \mathbb{S}^r_k$. If $r=n+1$, we denote $\nu^k_{n+1}:=\nu^k_{n,n+1}$. By the very definition, for each $E\in \mathcal{B}(\mathbb{S}_k)$ and $\mathbf{b}^k_{n}\in \mathbb{S}_k^{n}$, we have

\begin{equation}\label{preTP}
\nu^k_{n+1}(E|\mathbf{b}^k_{n})= \mathbb{P}\Big\{(\Delta T^k_{n+1}, \eta^k_{n+1})\in E|\mathcal{A}^k_{n} = \mathbf{b}^k_{n}\Big\}; n\ge 1.
\end{equation}
In other words, $\nu^k_{n+1}$ is the transition probability of the discrete-type skeleton $\mathscr{D}$ from step $n$ to $n+1$.

Next, we present a closed-form expression for the transition kernel (\ref{preTP}). For a given $\mathbf{b}^k_n = (s^k_1,\tilde{i}^k_1,\ldots, s^k_n, \tilde{i}^k_n)$, we define

\begin{equation}\label{pfunction}
\wp_\lambda(\textbf{b}^k_{n}):=\max\{1\le j\le n; \aleph_1(\tilde{i}^k_j)=\lambda\},
\end{equation}
where in (\ref{pfunction}), we make the convention that $\max\{\emptyset\}=0$. Moreover,  For a given $\mathbf{b}^k_n \in \mathbb{S}^k_n$, we set $\textbf{i}^k_{n}:=\big(\tilde{i}^k_1,\ldots, \tilde{i}^k_{n}\big)$ and we define

$$\mathbb{j}_{\lambda}(\textbf{i}^k_n):=\{\text{Number of jumps in the vector}~\textbf{i}^k_{n}~\text{which occurs in the $\lambda$ coordinate}\}.$$
For instance, if $\textbf{i}^k_3 = \big((-1,0), (0,1), (1,0)\big)$, then $\mathbb{j}_{1}(\textbf{i}^k_3)= 2$. For each $\textbf{b}^k_n\in \mathbb{S}^n_k$, we set

\begin{equation}\label{tknfunction}
t^k_n(\textbf{b}^k_n) : = \sum_{\beta=1}^n s^k_\beta.
\end{equation}
We then define

\begin{equation}\label{tkmodfunction}
t^{k,\lambda}_{\mathbb{j}_{\lambda}}(\textbf{b}^k_{n}): = \sum_{\beta=1}^{\wp_\lambda(\textbf{b}^k_{n})}s^k_\beta
\end{equation}
for $\lambda\in \{1,\ldots, d\}$ and $\textbf{b}^k_n\in \mathbb{S}^n_k$. We set $t^{k,\lambda}_0 = t^k_0 = 0$ and

$$\Delta^{k,\lambda}_n(\mathbf{b}^k_n):=t^k_n(\mathbf{b}^k_n)  - t^{k,\lambda}_{\mathbb{j}_{\lambda}}(\mathbf{b}^k_n).$$

When no confusion arises, we omit the dependence on the variable $\textbf{b}^k_n$ in $t^k_n$, $t^{k,\lambda}_{\mathbb{j}_{\lambda}}$ and $\Delta^{k,\lambda}_n$. Let $f_k$ be the density of the hitting time $T^{k,1}_1$ (see e.g Section 5.3 in \cite{milstein}). We make use of the information set described in (\ref{pfunction}), (\ref{tknfunction}) and (\ref{tkmodfunction}). We define

$$E_{j}:=\{1,2,\dots,j-1,j+1,\dots,d\}$$
and
$$f_{min}^k(t):=f_{min}^k(\textbf{b}^k_{n},j,t):=\displaystyle\left( \prod_{\lambda\in E_{j}} f_k\big(t+\Delta^{k,\lambda}_n\big)\right),$$
for $(\textbf{b}^k_{n},j,t)\in \mathbb{S}^{n-1}_k\times \{1, \ldots, d\}\times \mathbb{R}_+$
For instance,  if $d=2$, we have
$$f^k_{min}(\textbf{b}^k_n,j,t) = \left\{
\begin{array}{rl}
f_k\big(t+\Delta^{k,2}_n(\textbf{b}^k_n) \big); & \hbox{if} \ j=1 \\
f_k\big(t+\Delta^{k,1}_n(\textbf{b}^k_n) \big);& \hbox{if} \ j=2 \\
\end{array}
\right.
$$
for $(\textbf{b}^k_n,j,t)\in \mathbb{S}^n_k\times \{1,2\}\times \mathbb{R}_+.$


\begin{theorem}\label{disintegrationTH}
For each $\mathbf{b}^k_{n}\in\mathbb{S}^{n}_k$, $(j,\ell)\in \{1,\ldots, d\}\times \{-1,1\}$ and $-\infty< a < b < +\infty$, we have

\begin{equation}\label{disinformula}
\nu^k_{n+1}\big( (a,b)\times \aleph^{-1}(\{j,\ell\}) |\mathbf{b}^k_n\big) = \left\{
\begin{array}{rl}
\frac{1}{2}\int_a^bf_k(s)ds; & \hbox{if} \ d=1 \\
\frac{1}{2}L_{k,n}\big(j,(a,b)\big) \times I_{k,n}(j);& \hbox{if} \ d>1 \\
\end{array}
\right.
\end{equation}
where

$$L_{k,n}\big( j,(a,b)  \big):=\frac{\int_{a+\Delta^{k,j}_n}^{b+\Delta^{k,j}_n}f_k(x)dx}{\int_{\Delta^{k,j}_n}^{+\infty}f_k(x)dx}$$

$$I_{k,n}(j):=\frac{\int_{-\infty}^{0}\int_{-s}^\infty f_{k}\big(s+t+\Delta^{k,j}_n\big)f^k_{min}(\mathbf{b}^k_n,j,t)dtds}{\prod_{\lambda=1}^d \int_{\Delta^{k,\lambda}_n}^{+\infty} f_{k}(t)dt}.
$$

\end{theorem}
The proof of Theorem \ref{disintegrationTH} is postponed to the Appendix section. The importance of this formula lies on the following representation for conditional expectations: If $G = \Phi(\mathcal{A}^k_m)$ for a Borel function $\Phi:\mathbb{S}_k^m\rightarrow \mathbb{R}$ with $m\ge 1$, then

\begin{equation}\label{condEXP}
\mathbb{E}\Big[G\big| \mathcal{F}^k_{T^k_{m-1}}\Big] = \int_{\mathbb{S}_k} \Phi(\mathcal{A}^k_{m-1},s^k_m,\tilde{i}^k_m)\nu^k_m(ds^k_md\tilde{i}^k_m|\mathcal{A}^k_{m-1})~a.s.
\end{equation}
and iterating conditional expectations, formula (\ref{condEXP}) allows us to compute $\mathbb{E}\Big[G\big| \mathcal{F}^k_{T^k_{j}}\Big]$ for each $j=0,\ldots, m-1$.

\section{Dynamic programming in non-Markovian optimal stopping problems}\label{DPPsection}
In order to make this paper self-contained, we briefly recall the main results given by Le\~ao, Ohashi and Russo \cite{LEAO_OHASHI2017.3}. Let us denote $\textbf{B}^p(\mathbb{F})$ as the space of c\`adl\`ag $\mathbb{F}$-adapted processes $X$ such that

$$\|X\|^p_{\textbf{B}^p}:=\mathbb{E}\sup_{0\le t\le T}|X(t)|^p< \infty$$
where $1\le p< \infty$.
Throughout this work, we are going to fix a terminal time $0 < T < +\infty$. The following concept will be important in this work.

\begin{definition}\label{GASdef}
We say that $\mathcal{Y} = \big((X^k)_{k\ge 1},\mathscr{D}\big)$ is an \textbf{imbedded discrete structure} for $X$ if $X^k$ is a sequence of $\mathbb{F}^k$-adapted pure jump processes of the form

\begin{equation}\label{purejumpmodel}
X^k(t) = \sum_{n=0}^\infty X^{k}(T^k_n)1\!\!1_{\{T^k_n\le t < T^k_{n+1}\}}; 0\le t\le T,
\end{equation}
it has integrable quadratic variation $\mathbb{E}[X^k,X^k](T)< \infty; k\ge 1$, and

\begin{equation}\label{scdef}
\lim_{k\rightarrow+\infty}\|X^k-X\|^p_{\textbf{B}^p}=0
\end{equation}
for some $p\ge 1$.
\end{definition}

For $t\le T$, we denote $\mathcal{T}_t(\mathbb{F})$ as the set of all $\mathbb{F}$-stopping times $\tau$ such that $t\le \tau\le T$~a.s. For an integer $n\geq 0$, we denote by $\mathcal{T}_{k,n}(\mathbb{F}):=\mathcal{T}_{t}(\mathbb{F})$ for $t=T^k_n$. To shorten notation, we set $\mathcal{T}_{k,n}:=\mathcal{T}_{k,n}(\mathbb{F})$. Throughout this article, we assume that the underlying reward process $Z$ is an $\mathbb{F}$-adapted continuous process and it satisfies the following integrability condition:

\

\noindent \textbf{(A1)} $\|Z\|^p_{\textbf{B}^p}< \infty~\forall p\ge 1.$

\

For a given reward process $Z$, let $S$ be the Snell envelope associated with $Z$

\[
S (t):= \text{ess} \sup_{\tau\in \mathcal{T}_t(\mathbb{F})} \mathbb{E} \left[ Z(\tau)  \mid \mathcal{F}_t \right], \quad 0 \leq t \leq T.
\]
We assume $S$ satisfies the following integrability condition:

\

\noindent \textbf{(A2)} $\|Z\|^p_{\textbf{B}^p}< \infty~\forall p\ge 1.$

\

Since the optimal stopping time problem at hand takes place on the compact set $[0,T]$, it is crucial to know the correct number of periods in our discretization scheme. For this purpose, let us denote $\lceil x\rceil$ as the smallest natural number bigger
or equal to $x\ge 0$. We recall (see Lemma 3.1 in Le\~ao, Ohashi and Russo \cite{LEAO_OHASHI2017.3}) that if $e(k,t):=d\lceil 2^{2k}t \rceil$, then for each $t\ge 0$

$$T^k_{e(k,t)}\rightarrow t$$
almost surely and in $L^2(\mathbb{P})$ as $k\rightarrow+\infty$. Due to this result, we will reduce the analysis to the deterministic number of periods $e(k,T)$.

\begin{remark}\label{Markov_A}
For each $k\ge 1$, one can view

$$\mathcal{A}^k_n; n=1,\ldots, e(k,T)$$
as a discrete-time Markov process on the enlarged state space $\mathbb{S}^{e(k,T)}_k$ by setting

$$\mathcal{A}^k_n=(\Delta T^k_1, \eta^k_1, \ldots, \Delta T^k_n,\eta^k_n, \Delta T^k_n,\eta^k_n,\ldots, \Delta T^k_n,\eta^k_n)\in \mathbb{S}_k^{e(k,T)}$$
for each $n=1, \ldots, e(k,T)$.
\end{remark}

We denote $D^{k,m}_n$ as the set of all $\mathbb{F}^k$-stopping times of the form

\begin{equation}\label{form1}
\tau = \sum_{i=n}^{m}T^k_i1\!\!1_{\{\tau = T^k_i\}}
\end{equation}
where $\{\tau = T^k_i; i=n, \ldots, m\}$ is a partition of $\Omega$ and $0\le n \le m$. Let us denote $D^k_{n,T} := \{\eta\wedge T; \eta\in D^{k,\infty}_n\}$.

Let $\{Z^k; k\ge 1\}$ be a sequence of pure jump processes of the form (\ref{purejumpmodel}) and let $\{S^k; k\ge 1\}$ be the associated value process given by

\begin{equation}\label{discretevaluep}
S^k(t) := \sum_{n=0}^{e(k,T)} S^k(T^k_n)\mathds{1}_{\{T^k_n\le t < T^k_{n+1}\}}; 0\le t\le T,
\end{equation}
where

$$S^k(T^k_n):= \esssup_{\tau\in D^{k,e(k,T)}_{n}}\mathbb{E}\Big[ Z^k(\tau\wedge T)\big|\mathcal{F}^k_{T^k_n}\Big];0\le n\le e(k,T).$$
In the sequel, we denote

$$U^{\mathcal{Y},k,\textbf{p}}S(T^k_i):=\mathbb{E}\Bigg[\frac{\Delta S^{k}(T^k_{i+1})}{\epsilon_k^2}\Big|\mathcal{F}^k_{T^k_i}\Bigg]; 0\le i\le e(k,T)-1.$$
Let us now recall two major results presented in Le\~ao, Ohashi and Russo \cite{LEAO_OHASHI2017.3}, namely Theorems 3.1 and 3.2.

\

\noindent \textbf{Theorem 3.1 in Le\~ao, Ohashi and Russo~\cite{LEAO_OHASHI2017.3}}. Let $S$ be the Snell envelope associated with a reward process $Z$ satisfying (A1-A2). Let $\{Z^k; k\ge 1\}$ be a sequence of pure jump processes of the form (\ref{purejumpmodel}) and let $\{S^k; k\ge 1\}$ be the associated value process given by (\ref{discretevaluep}). If $\mathcal{Z} = \big( (Z^k)_{k\ge 1},\mathscr{D}\big)$ is an imbedded discrete structure for $Z$ where (\ref{scdef}) holds for $p>1$, then $\mathcal{S} = \big((S^k)_{k\ge 1},\mathscr{D}\big)$ is an imbedded discrete structure for $S$ where

$$\lim_{k\rightarrow+\infty}\mathbb{E}\sup_{0\le t\le T}|S^k(t) - S(t)|=0.$$
Moreover, $\{S^k; k\ge 1\}$ is the unique pure jump process of the form (\ref{purejumpmodel}) which satisfies the following variational inequality

\begin{eqnarray}\label{varineq}
\max \Big\{ U^{\mathcal{Y},k,\textbf{p}}S(T^k_i);  Z^k(T^k_i\wedge T) - S^{k}(T^k_i) \Big\} & = & 0\quad i=e(k,T)-1, \ldots, 0,~a.s. \\
\nonumber S^{k} (T^k_{e(k,T)}) &=&Z^k (T^k_{e(k,T)}\wedge T)~a.s.
\end{eqnarray}

\

Next, for sake of completeness, we recall some explanations given by \cite{LEAO_OHASHI2017.3} which will be important to the Longstaff-Schwartz algorithm given in the next section. The variational inequality is equivalent to

\begin{eqnarray}
\nonumber S^{k}(T^k_n)&=&\max\Big\{Z^k(T^k_n\wedge T); \mathbb{E}\Big[S^{k}(T^k_{n+1})\big|\mathcal{F}^k_{T^k_n}\Big]\Big\};~n=e(k,T)-1, e(k,T)-2, \ldots, 0~a.s.\\
\label{qeq}& &\\
\nonumber S^{k}(T^k_{e(k,T)})&=& Z^k(T^k_{e(k,T)}\wedge T)~a.s.
\end{eqnarray}
For each $n\in \{0,\ldots, e(k,T)\}$, there exist Borel-measurable functions $\mathbb{V}^k_n:\mathbb{S}^n_k\rightarrow\mathbb{R}$ and $\mathbb{Z}^k_n:\mathbb{S}^n_k\rightarrow\mathbb{R}$ which realize

\begin{equation}\label{listpathwise}
S^{k}(T^k_n) = \mathbb{V}^k_n(\mathcal{A}^k_n)~a.s.\quad \text{and}~Z^k(T^k_n\wedge T) = \mathbb{Z}^k_n(\mathcal{A}^k_n)~a.s.; n=0, \ldots, e(k,T).
\end{equation}
Moreover, the sequence $\mathbb{V}^k_i:\mathbb{S}^i\rightarrow \mathbb{R}; 0\le i\le e(k,T)$ are determined by the following dynamic programming algorithm

\begin{eqnarray}
\nonumber  \mathbb{V}^k_i(\textbf{b}^k_i)&=&\max\Big\{\mathbb{Z}^{k}_i(\textbf{b}^k_i); \mathbb{E}\big[\mathbb{V}^k_{i+1}(\mathcal{A}^{k}_{i+1})|\mathcal{A}^{k}_i =\textbf{b}^k_{i}\big]\Big\};~0\le i\le e(k,T)-1\\
\label{DPA1} \mathbb{V}^k_{e(k,T)}(\textbf{b}^k_{e(k,T)}) &=& \mathbb{Z}^{k}_{e(k,T)}(\textbf{b}^k_{e(k,T)}),
\end{eqnarray}
for each $\textbf{b}^k_{i}\in \mathbb{S}^i_k$ for $0\le i\le e(k,T)$ and $k\ge 1$.
The dynamic programming algorithm allows us to define the stopping and continuation regions as follows

$$\textbf{S}(i,k):=\Big\{\textbf{b}^k_{i}\in  \mathbb{S}^{i}_k; \mathbb{Z}^k_i(\textbf{b}^k_i) = \mathbb{V}^k_i(\textbf{b}^k_i) \Big\}\quad \text{(stopping region)}$$

$$\textbf{D}(i,k):=\Big\{\textbf{b}^k_{i}\in  \mathbb{S}^{i}_k; \mathbb{V}^k_i(\textbf{b}^k_i) > \mathbb{Z}^k_i(\textbf{b}^k_i) \Big\}\quad \text{(continuation region)}$$
where $0\le i\le e(k,T)$. For a given non-negative integer $n\ge 0$, let us denote $J^{k,e(k,T)}_n$ as the set of all $(\mathcal{F}^k_{T^k_i})^{e(k,T)}_{i=n}$-stopping times $\eta$ having the form

\begin{equation}\label{isom}
\eta = \sum_{i=n}^{e(k,T)} i1\!\!1_{\{\tau=i\}},
\end{equation}
where $\{\tau=i\}; n\le i\le e(k,T)$ constitutes a partition of $\Omega$. Clearly, there exists a natural isomorphism between $J^{k,e(k,T)}_n$ and $D^{k,e(k,T)}_n$. Let us define

\begin{equation}\label{isom1}
Y^k(i):=Z^k(T^k_i\wedge T); i\ge 0.
\end{equation}
By construction,

$$\esssup_{\eta\in J^{k,e(k,T)}_n}\mathbb{E}\big[Y^k(\eta)|\mathcal{F}^k_{T^k_n}\big] = S^k(T^k_n)~a.s.$$
for each $0\le n \le e(k,T)$.

The smallest $(\mathcal{F}^k_{T^k_i})^{e(k,T)}_{i=0}$-optimal stopping-time w.r.t. the problem

$$\sup_{\tau \in J^{e(k,T)}_0}\mathbb{E} \big[Y^k(\tau)\big] = \sup_{\eta\in D^{k,e(k,T)}_0}\mathbb{E}\big[Z^k(\eta\wedge T)\big]$$
is given by
\begin{eqnarray}
\nonumber \tau^{k}&:=&\min \Big\{0\le j\le e(k,T); \mathcal{A}^k_{j}\in \textbf{S}(j,k)\Big\}\\
\label{bbbb}& &\\
\nonumber&=&\min\Big\{0\le j\le e(k,T); S^{k}(T^k_j) = Z^k(T^k_j \wedge T)\Big\}
\end{eqnarray}
which is finite a.s. by construction. Moreover, the dynamic programming principle can be written as

\begin{equation}\label{DPST}
\left\{\begin{array}{l}
 \tau^{k}_{e(k,T)}:= e(k,T) \\
\tau^{k}_{j}:= j 1\!\!1_{G^k_j} + \tau^{k}_{j+1}1\!\!1_{(G^{k}_j)^c}; 0\le j \le e(k,T)-1
\end{array}\right.
\end{equation}
where

$$G^k_j: = \Bigg\{\mathbb{Z}^k_j (\mathcal{A}^k_j)\ge \mathbb{E}\Big[\mathbb{Z}^k_{\tau^{k}_{j+1}}(\mathcal{A}^k_{\tau^{k}_{j+1}})\big|\mathcal{A}^k_j\Big]\Bigg\}; 0\le j\le e(k,T)-1$$
and $\tau^{k}=\tau^{k}_0$ a.s. The sequence of functions $\mathbf{U}^k_j:\mathbb{S}^j_k\rightarrow\mathbb{R}$

\begin{equation}\label{cvalues}
\textbf{b}^k_j\mapsto \mathbf{U}^k_j(\textbf{b}^k_j):=\mathbb{E}\Big[\mathbb{Z}^k_{\tau^{k}_{j+1}}(\mathcal{A}^k_{\tau^{k}_{j+1}})\big|\mathcal{A}^k_j=\textbf{b}^k_j\Big]; 0\le j\le e(k,T)-1
\end{equation}
are called \textit{continuation values}.

The value functional which gives the best payoff can be reconstructed by means of our dynamic programming principle over the $e(k,T)$-steps in such way that

\begin{equation}\label{valuef}
V^k_0:=\sup_{\eta\in J^{k,e(k,T)}_0}\mathbb{E}\big[Y^k(\eta)\big] = \mathbb{V}^k_0(\textbf{0}) = \max\Big\{\mathbb{Z}^k_0(\textbf{0}); \mathbb{E}\big[\mathbb{V}^k_1 (\mathcal{A}^k_1)\big]\Big\},
\end{equation}
where $\mathbb{E}\big[\mathbb{V}^k_1 (\mathcal{A}^k_1)\big] = \mathbb{E}\Big[\mathbb{Z}^k_{\tau^{k}_1}(\mathcal{A}^k_{\tau^{k}_1})\Big] = \mathbf{U}^k_0(\textbf{0})$. Moreover,

\begin{eqnarray}
\label{aaa} \mathbb{E}\big[Y^k(\tau^k)\big]&=&\mathbb{E}\big[Z^k(T^k_{\tau^{k}}\wedge T)\big] = \sup_{\tau\in D^k_{0,T}} \mathbb{E}\big[Z^k(\tau \wedge T^k_{e(k,T)})\big]\\
\nonumber & &\\
\label{aaaa}&=&\sup_{\tau\in \mathcal{T}_0(\mathbb{F})}\mathbb{E} \big[Z^k(\tau \wedge T^k_{e(k,T)})\big]
\end{eqnarray}
where identity (\ref{aaaa}) is due to Proposition 3.1 in \cite{LEAO_OHASHI2017.3}. We now finish this section by recalling the following result. In the sequel, we recall that $S(0) = \sup_{\eta\in \mathcal{T}_0(\mathbb{F})}\mathbb{E}\big[Z(\eta)\big]$.

\

\noindent \textbf{Theorem 3.2 in Le\~ao, Ohashi and Russo~\cite{LEAO_OHASHI2017.3}.} If $\mathcal{Z} = \big((Z^k)_{k\ge 1}, \mathscr{D}\big)$ is an imbedded discrete structure for the reward process $Z$, then $T^k_{\tau^{k}}\wedge T$ is an $\epsilon$-optimal stopping time in the Brownian filtration, i.e., for a given $\epsilon>0$,

$$\sup_{\eta\in \mathcal{T}_0(\mathbb{F})}\mathbb{E}\big[Z(\eta)\big] -\epsilon< \mathbb{E}\big[Z(T^k_{\tau^k}\wedge T)\big]$$
for every $k$ sufficiently large. Moreover,

\begin{equation}\label{abcov}
\Big|\sup_{\tau\in \mathcal{T}_0(\mathbb{F})}\mathbb{E}\big[Z(\tau)\big] -  V^k_0\Big|\le \|Z^k(\cdot\wedge T^k_{e(k,T)}) - Z\|_{\mathbf{B}^1}\rightarrow 0
\end{equation}
as $k\rightarrow+\infty$.

\

Let us now present a Longstaff-Schwartz algorithm for the value functions

$$V^k_0; k\ge 1.$$

\subsection{Longstaff-Schwartz algorithm}
Each random element $\mathcal{A}^k_n$ induces an image probability measure $\rho^k_n:=\mathbb{P}^k_n$ on $\mathbb{S}^n_k; n\ge 1$ where $\rho_0$ is just the Dirac concentrated on $\textbf{0}$. For each $m=0, \ldots e(k,T)-1$, the family $\{\mathcal{A}^k_{j}; j=m, \ldots, e(k,T)\}$ induces an image probability measure $\rho^k_{m,e(k,T)}$ on the $e(k,T) - m+1$-fold cartesian product space

$$(\underbrace{\mathbb{S}^{e(k,T)}_k\times \dots\times \mathbb{S}^{e(k,T)}_k}_{e(k,T)-m+1}).$$

Throughout this article, we assume that $\mathbb{Z}^k_n:\mathbb{S}^n_k\rightarrow \mathbb{R} \in L^2(\mathbb{S}^n_k,\rho^k_n)$ for every $n=0, \ldots, e(k,T)$. Let us now select a subset $\{\widehat{\mathbf{U}}^k_j;j=0, \ldots e(k,T)-1 \}$  of functions such that $\widehat{\mathbf{U}}^k_j\in L^2(\mathbb{S}^{j}_k,\rho^k_j)$ for each $j=0, \ldots, e(k,T)-1$. For each choice of functions, we set inductively

\begin{equation}\label{DPST1}
\left\{\begin{array}{l}
 \widehat{\tau}^{k}_{e(k,T)}:= e(k,T) \\
\widehat{\tau}^{k}_{j}:= j 1\!\!1_{\widehat{G}^k_j} + \tau^{k}_{j+1}1\!\!1_{(\widehat{G}^{k}_j)^c}; 0\le j \le e(k,T)-1
\end{array}\right.
\end{equation}
where $\widehat{G}^k_j: = \{\mathbb{Z}^k_j (\mathcal{A}^k_j)\ge \widehat{\mathbf{U}}^{k}_j(\mathcal{A}^k_j)\}; 0\le j\le e(k,T)-1$ and $\widehat{\tau}^{k}=\widehat{\tau}^{k}_0$. Here, $\widehat{\mathbf{U}}^k_j(\cdot)$ should be interpreted as a suitable approximation of $\mathbb{E}\big[\mathbb{Z}^k_{\widehat{\tau}^{k}_{j+1}}(\mathcal{A}^k_{\widehat{\tau}^{k}_{j+1}})|\mathcal{A}^k_j=\cdot\big]$ for each $j=0, \ldots, e(k,T)-1$.

The set $\{\widehat{\tau}^{k}_j; 0\le j \le e(k,T)\}$ induces a set of conditional expectations $$\mathbb{E}\Big[Y^k(\widehat{\tau}^{k}_{j+1})|\mathcal{A}^k_{j}\Big]= \mathbb{E}\Big[\mathbb{Z}^k_{\widehat{\tau}^{k}_{j+1}}(\mathcal{A}^k_{\widehat{\tau}^{k}_{j+1}})\big|\mathcal{A}^k_j\Big]; j=0, \ldots, e(k,T)-1$$
so that one can postulate

$$\widehat{V}^k_0:=\max\big\{\mathbb{Z}^k_0(\textbf{0}); \widehat{\mathbf{U}}^{k}_0(\textbf{0})\big\}$$
as a possible approximation for (\ref{valuef}). Inspired by Zanger \cite{zanger}, the Monte Carlo algorithm is given by the following lines.

\

\noindent \textbf{Longstaff-Schwartz Algorithm}. Let us fix $k\ge 1$. Step (0): Given any positive integer $N$, select $\mathcal{H}^k_{N,0}\subset \mathbb{R}$. For each $j=1, \ldots, e(k,T)-1$, select $\mathcal{H}^k_{N,j}\subset L^2(\mathbb{S}^j_k, \rho^k_j)$. The sets $\mathcal{H}^k_{N,j}; 0\le j\le e(k,T)$ possibly depends on $N$ and the choice is dictated by some a priori information that one has about the continuation values (\ref{cvalues}). In learning theory literature, they are usually called \textit{approximation architectures}. From $\Big(\mathcal{A}^k_\ell; 0\le \ell \le e(k,T)\Big)$, generate $N$ independent samples $\mathcal{A}^k_{0,i}, \mathcal{A}^k_{1,i}, \ldots, \mathcal{A}^k_{e(k,T),i}; i=1, \ldots, N$. For each $\ell=0, \ldots e(k,T)$, let us denote

$$\mathbf{A}^k_{\ell N}:=\big(\mathcal{A}^k_{\ell,1},\mathcal{A}^k_{\ell,2}, \ldots, \mathcal{A}^k_{\ell,N}; \ldots; \mathcal{A}^k_{e(k,T), 1},\mathcal{A}^k_{e(k,T), 2}, \ldots, \mathcal{A}^k_{e(k,T),N} \big)$$
with $(e(k,T)-\ell+1)N$-factors.

Step (1): For $j=e(k,T)-1$, we set $\widehat{\tau}^k_{j+1} := \widehat{\tau}^k_{j+1}(\mathbf{A}^k_{(j+1)N}) := e(k,T)$ and generate $\{  (\mathcal{A}^k_{j,i}, (\mathbb{Z}^k_{\widehat{\tau}^k_{j+1}})_i); 1\le i\le N\}$, where we define $(\mathbb{Z}^k_{\widehat{\tau}^k_{j+1}})_i:=\mathbb{Z}^k_{e(k,T)}(\mathcal{A}^k_{e(k,T),i}); 1\le i\le N.$ We then select

\begin{equation}\label{mini}
\widehat{\mathbf{U}}^k_{j} := \text{arg min}_{g\in \mathcal{H}^k_{N,j}}\frac{1}{N}\sum_{i=1}^N \Big((\mathbb{Z}^k_{\widehat{\tau}^k_{j+1}})_i - g(\mathcal{A}^k_{j,i})     \Big)^2.
\end{equation}
One should notice that $\widehat{\mathbf{U}}^k_{j}$ is a functional of $\mathbf{A}^k_{j N}$ so that we assume the existence of a minimizer \begin{equation}\label{minimizerMC}
\widehat{\mathbf{U}}^k_{j}:\mathbb{S}^{j}_k\times \big(\mathbb{S}^{j}_k\big)^N \times \ldots \times \big(\mathbb{S}^{e(k,T)}_k\big)^N\rightarrow \mathbb{R}
\end{equation}
of (\ref{mini}) which possibly can depend on $N$. With $\widehat{\mathbf{U}}^k_j$ at hand, we compute $\widehat{\tau}^k_{ji} = \Big( \widehat{\tau}^k_j \big( \mathbf{A}^k_{jN} \big) \Big)_i$, the value that $\widehat{\tau}^k_{j} =\widehat{\tau}^k_{j}\big( \mathbf{A}^k_{jN} \big) $ assumes based on the $i$-th sample according to (\ref{DPST}), i.e., we define

\begin{equation}\label{DPST2}
\widehat{\tau}^{k}_{ji}:= j 1\!\!1_{\big\{\mathbb{Z}^k_j(\mathcal{A}^k_{j,i})\ge \widehat{\mathbf{U}}^k_j(\mathcal{A}^k_{j,i},\mathbf{A}^k_{jN})\big\}} + \widehat{\tau}^k_{(j+1)i} 1\!\!1_{\big\{\mathbb{Z}^k_j(\mathcal{A}^k_{j,i})< \widehat{\mathbf{U}}^k_j(\mathcal{A}^k_{j,i},\mathbf{A}^k_{jN})\big\}}
\end{equation}
for $1\le i\le N$. In this case, we set

\begin{equation}\label{DPST3}
\Big(\mathbb{Z}^k_{\widehat{\tau}^k_j}\Big)_i:=\left\{
\begin{array}{rl}
\mathbb{Z}^k_j\big(\mathcal{A}^k_{j,i}\big); & \hbox{if} \ \widehat{\tau}^k_{ji}=j \\
\mathbb{Z}^k_{\widehat{\tau}^k_{(j+1)i}}\big(\mathcal{A}^k_{\widehat{\tau}^k_{(j+1)i},i}\big);& \hbox{if} \ \widehat{\tau}^k_{ji} = \widehat{\tau}^k_{(j+1)i}\\
\end{array}
\right.
\end{equation}
where $\widehat{\tau}^k_{(j+1)i}=e(k,T)$ for $1\le i\le N$.

Step (2): Based on (\ref{mini}), (\ref{DPST2}) and (\ref{DPST3}), we then repeat this procedure inductively $j=e(k,T)-2, \ldots, 1, 0$ until step $j=0$ to get

$$\Big( \widehat{\tau}^k_{ji}, \widehat{\mathbf{U}}^k_{j}, \big(\mathbb{Z}^k_{\widehat{\tau}^k_{j}}\big)_i   \Big); 0\le j\le e(k,T), 1\le i\le N.$$

Step (3): For $j=0$, we set

$$\widehat{V}_0(\mathbf{A}^k_{0N}) := \text{max}\Big\{\mathbb{Z}^k_0(\textbf{0}), \widehat{\mathbf{U}}^k_0(\mathbf{A}^k_{0N})\Big\}.$$


\section{Error estimates for the Monte Carlo method}\label{errorESTSECTION}
In this section, we present the error estimates for the Monte Carlo method described in previous section. Clearly, the most important object in the Monte Carlo scheme is the computation of the underlying conditional expectations. For this purpose, we make use of some machinery from statistical learning theory having the regression function as the fundamental object. Throughout this section, we fix a list a functions $\mathbb{Z}^k_n:\mathbb{S}^n_k\rightarrow\mathbb{R}; n=0,\ldots, e(k,T)$ realizing (\ref{listpathwise}) for a given imbedded discrete structure $\mathcal{Z} = \big((Z^k)_{k\ge 1},\mathscr{D}\big)$ associated with a given reward process $Z$. For concrete examples of pathwise representations of structures $(\mathbb{Z}^k_n)_{n=0}^{e(k,T)}$, we refer the reader to Section 5 of Le\~ao, Ohashi and Russo \cite{LEAO_OHASHI2017.3} and (\ref{Zpathwise}).

\subsection{A bit of learning theory}\label{learningsec}
In the sequel, let $\gamma$ be a Borel probability measure on a product space $\mathcal{X}\times \mathbb{R}$ where we assume that $\mathcal{X}$ is a Polish space equipped with the Borel sigma algebra $\mathcal{B}(\mathcal{X})$. Let $\gamma_{\mathcal{X}}(A):=\gamma(A\times \mathbb{R}); A\in \mathcal{B}(\mathcal{X})$ be the marginal probability measure onto $\mathcal{X}$. Since $\mathcal{X}\times \mathbb{R}$ is Polish, one can disintegrate $\gamma$ w.r.t. $\gamma_{\mathcal{X}}$ so that there exists a unique (up to $\gamma_\mathcal{X}$-null sets) measure-valued function $x\mapsto \gamma_x(dr)$
which realizes

$$\int_{\mathcal{X}\times \mathbb{R}}f(x,r)\gamma(dx,dr) = \int_{\mathcal{X}}\int_{\mathbb{R}}f(x,r)\gamma_x(dr)\gamma_{\mathcal{X}}(dx)$$
for every $f\in L^1(\mathcal{X}\times\mathbb{R}; \gamma)$. The function

$$f_\gamma(x):=\int_{\mathbb{R}}r\gamma_x(dr); x\in \mathcal{X}$$
is the regression function of $\gamma$. The most important object to minimize is the so-called risk functional

$$\mathcal{T}(f):=\int_{\mathcal{X}\times\mathbb{R}}|r-f(x)|^2\gamma(dx,dr)$$
where $f:\mathcal{X}\rightarrow\mathbb{R} \in L^2(\mathcal{X}; \gamma_\mathcal{X})$. Since $\int_{\mathbb{R}}(f_\gamma(x) - r)\gamma_x(dr) = 0; x\in \mathcal{X}$, then one can easily check that

\begin{equation}\label{frisk}
\mathcal{T}(f) = \mathcal{T}(f_\gamma) + \|f-f_\gamma\|^2_{L^2(\mathcal{X};\gamma_{\mathcal{X}})}
\end{equation}
for every $f\in  L^2(\mathcal{X}; \gamma_{\mathcal{X}})$. Therefore,

$$f_\gamma = \text{argmin}_{\substack{f\in L^2(\mathcal{X}; \gamma_\mathcal{X})}}\mathcal{T}(f).$$
One central problem in statistical learning theory is to approximate $f_\gamma$ by means of some approximation architecture $\mathcal{H}\subset L^2(\mathcal{X};\gamma_\mathcal{X})$. Identity (\ref{frisk}) yields

$$\text{arg min}_{g\in \mathcal{H}}\|g-f_\gamma\|^2_{L^2(\gamma_\mathcal{X})} = \text{arg min}_{g\in\mathcal{H}}\int_{\mathcal{X}\times\mathbb{R}}\big(g(x) - r \big)^2\gamma(dx,dr).$$

In this article, we apply learning theory in the following setup: At first, for a given set of continuation values $\mathbf{U}^k_\ell:\mathbb{S}^\ell_k\rightarrow \mathbb{R}; \ell=0, \ldots, e(k,T)-1$,  we select a list of approximating functions $\widehat{\mathbf{U}}^k_\ell:\mathbb{S}^\ell_k\rightarrow \mathbb{R}\in L^2(\rho^k_\ell); \ell=0, \ldots, e(k,T)-1$. With such approximating functions, for each $j=0, \ldots, e(k,T)-1$, we define a mapping $\widehat{\varphi}^k_j:\otimes_{\ell=j}^{e(k,T)} \mathbb{S}^\ell_k \rightarrow \mathbb{S}^j_k \times \mathbb{R}$ as follows

\begin{equation}\label{phifunc}
\big(\textbf{b}^k_j,\textbf{b}^k_{j+1}, \ldots, \textbf{b}^k_{e(k,T)}\big)\mapsto \Big(\textbf{b}^k_j,\mathbf{Z}^k_{\widehat{\tau}^k_{j+1}}\big(\textbf{b}^k_{j+1}, \ldots, \textbf{b}^k_{e(k,T)} \big)\Big)
\end{equation}
where $\mathbf{Z}^k_{\widehat{\tau}^k_{j+1}}: \otimes_{\ell=j+1}^{e(k,T)} \mathbb{S}^\ell_k\rightarrow \mathbb{R}$ is defined inductively by means of the following procedure: We start from $j=e(k,T)-1$, where we set $\mathbf{Z}^k_{\widehat{\tau}^k_{e(k,T)}}(\textbf{b}^k_{e(k,T)}) := \mathbb{Z}^k_{e(k,T)}(\textbf{b}^k_{e(k,T)}), \textbf{b}^k_{e(k,T)}\in \mathbb{S}^{e(k,T)}_k$ and

\begin{equation}\label{fundmap1}
\mathbf{Z}^k_{\widehat{\tau}^k_{j+1}}(\textbf{b}^k_{j+1}, \ldots, \textbf{b}^k_{e(k,T)}):=\left\{
\begin{array}{rl}
\mathbb{Z}^k_{j+1}(\textbf{b}^k_{j+1}); & \hbox{if} \ \mathbb{Z}^k_{j+1}(\textbf{b}^k_{j+1})\ge \widehat{\mathbf{U}}^k_{j+1}(\mathbf{b}^k_{j+1}) \\
\mathbf{Z}^k_{\widehat{\tau}^k_{j+2}}(\textbf{b}^k_{j+2}, \ldots, \textbf{b}^k_{e(k,T)});& \hbox{if} \ \mathbb{Z}^k_{j+1}(\textbf{b}^k_{j+1}) < \widehat{\mathbf{U}}^k_{j+1}(\mathbf{b}^k_{j+1}) \\
\end{array}
\right.
\end{equation}
for $j=e(k,T)-2, \ldots, 0$. Similarly, for $j=e(k,T)-1$, we set $\mathbf{Z}^k_{\tau^k_{e(k,T)}}(\textbf{b}^k_{e(k,T)}) := \mathbb{Z}^k_{e(k,T)}(\textbf{b}^k_{e(k,T)}), \textbf{b}^k_{e(k,T)}\in \mathbb{S}^{e(k,T)}_k$ and

\begin{equation}\label{fundmap2}
\mathbf{Z}^k_{\tau^k_{j+1}}(\textbf{b}^k_{j+1}, \ldots, \textbf{b}^k_{e(k,T)}):=\left\{
\begin{array}{rl}
\mathbb{Z}^k_{j+1}(\textbf{b}^k_{j+1}); & \hbox{if} \ \mathbb{Z}^k_{j+1}(\textbf{b}^k_{j+1})\ge \mathbf{U}^k_{j+1}(\mathbf{b}^k_{j+1}) \\
\mathbf{Z}^k_{\tau^k_{j+2}}(\textbf{b}^k_{j+2}, \ldots, \textbf{b}^k_{e(k,T)});& \hbox{if} \ \mathbb{Z}^k_{j+1}(\textbf{b}^k_{j+1}) < \mathbf{U}^k_{j+1}(\mathbf{b}^k_{j+1}) \\
\end{array}
\right.
\end{equation}
for $j=e(k,T)-2, \ldots, 0$.
\begin{remark}\label{useful1}
From (\ref{fundmap2}), we observe that

\begin{equation}\label{fundmap3}
\mathbf{Z}^k_{\tau^k_{j+1}}(\mathcal{A}^k_{(j+1)}, \ldots, \mathcal{A}^k_{e(k,T)}) = \mathbb{Z}^k_{\tau^k_{j+1}}(\mathcal{A}^k_{\tau^k_{j+1}})~a.s.
\end{equation}
so if we set $\mathbf{U}^k_{e(k,T)}(\mathcal{A}^k_{e(k,T)}):=\mathbb{Z}^k_{e(k,T)}(\mathcal{A}^k_{e(k,T)})$, then (\ref{fundmap3}) yields

\begin{equation}\label{cvalues1}
\mathbf{U}^k_j(\mathcal{A}^k_j) = \mathbb{E}\big[\mathbb{Z}^k_{j+1}(\mathcal{A}^k_{j+1})\vee \mathbf{U}^k_{j+1}(\mathcal{A}^k_{j+1})|\mathcal{A}^k_j\big], j=e(k,T)-1, \ldots, 0.
\end{equation}
By observing that

\begin{equation}\label{splittingsigmaal}
\mathcal{F}^k_{T^k_n} = \sigma(\mathcal{A}^k_n) = \mathcal{F}^k_{T^k_{n-1}}\vee \sigma(\Delta T^k_n, \eta^k_n) = \sigma(\mathcal{A}^k_{n-1}, \Delta T^k_n, \eta^k_n); n,k\ge 1,
\end{equation}
one can easily check that
\begin{equation}\label{cvalues2}
\mathbf{U}^k_j(\mathbf{b}^k_j) = \int_{\mathbb{S}_k}\mathbb{Z}^k_{j+1}(\mathbf{b}^k_{j},s^k_{j+1},\tilde{i}^k_{j+1} )\vee \mathbf{U}^k_{j+1}(\mathbf{b}^k_j,s^k_{j+1},\tilde{i}^k_{j+1})\nu^k_{j+1}(ds^k_{j+1}d\tilde{i}^k_{j+1}|\mathbf{b}^k_j)
\end{equation}
for $\mathbf{b}^k_j\in \mathbb{S}^j_k$ and $j=e(k,T)-1, \ldots, 0$, where $\nu^k_{j+1}$ is the regular conditional probability given by (\ref{disinformula}) and Theorem \ref{disintegrationTH}. With a slight abuse of notation, we write $\mathbb{Z}^k_{j+1}(\mathbf{b}^k_{j},s^k_{j+1},\tilde{i}^k_{j+1} )$ and
$\mathbf{U}^k_{j+1}(\mathbf{b}^k_j,s^k_{j+1},\tilde{i}^k_{j+1})$ as versions which realize

$$\mathbb{Z}^k_{j+1}(\mathcal{A}^k_{j},\Delta T^k_{j+1}, \eta^k_{j+1}) = \mathbb{Z}^k_{j+1}(\mathcal{A}^k_{j+1}), \quad \mathbf{U}^k_{j+1}(\mathcal{A}^k_{j},\Delta T^k_{j+1},\eta^k_{j+1}) = \mathbf{U}^k_{j+1}(\mathcal{A}^k_{j+1}),$$
respectively, for each $j=e(k,T)-1, \ldots, 0$.
\end{remark}

From a theoretical perspective, both $\mathbb{V}^k$ and $\mathbf{U}^k$ they are equivalent since they provide the solution of the optimal stopping problem. However, from the numerical point of view, working with continuation values is better than value functions due to the appearance of expectation which allows us to work with smoother functions. This will be crucial to obtain concrete approximation architectures for the continuation values.

\begin{remark}\label{useful2}
The dynamic programming principle (\ref{DPA1}), (\ref{DPST}) and (\ref{fundmap1}) allow us to apply statistical learning techniques to general non-Markovian optimal stopping problems in full generality. In this case, $\mathcal{X}$ will be $\mathbb{S}^j_k$ and $\rho^k_{j,e(k,T)}\circ (\widehat{\varphi}^k_j)^{-1}$ will play the role of $\gamma$ for each $j=0, \ldots, e(k,T)-1$.
\end{remark}

One crucial aspect in the obtention of concrete error estimates for our Monte Carlo scheme is the so-called Vapnik-Chervonenkis dimension (henceforth abbreviated by VC dimension) of a set. For readers who are not familiar with this concept, we recall the definition. For a given $\epsilon>0$ and a bounded subset $A\subset \mathbb{R}^N$, let

$$m(A,\epsilon):=\{B\subset \mathbb{R}^N; B~\text{is finite and}~\forall a\in A, \exists b\in B; 1/N \|a-b\|_1 <\epsilon\}$$
where we set $\|a-b\|_1 = \sum_{i=1}^N|a_i-b_i|$. Based on this quantity. we define $\mathcal{N}(\epsilon,A)$ as the cardinality of the smallest subset belonging to $m(A,\epsilon)$. Let $\mathcal{G}$ be a family of uniformly bounded real-valued functions defined on some subset $\Sigma\subset \mathbb{R}^m$. For each $v = (v_1, \ldots, v_N)\in \Sigma^N$, we set $\mathcal{G}(v) = \Big\{ \big(g(v_1), \ldots, g(v_N)\big)\in \mathbb{R}^N; g\in \mathcal{G}\big)  \Big\}$ so it makes sense to consider $\mathcal{N}(\epsilon,\mathcal{G}(v)); v\in \Sigma^N$. Now, for a given list $\{a_1, \ldots, a_n\}\subset \Sigma; n\ge 1$, we say that $\mathcal{G}$ \textit{shatter} $\{a_1, \ldots, a_n\}\subset \Sigma$ if there exists $r=(r_1, \ldots, r_n)\in \mathbb{R}^n$ such that for every $b=(b_1, \ldots, b_n)\in \{0,1\}^n$, there exists $g\in \mathcal{G}$ such that for each $i=1, \ldots, n$ $g(a_i) \ge r_i$ if $b_i=1$ and $g(a_i) < r_i$ if $b_i=0$. We are finally able to define

$$vc(\mathcal{G}) = \sup\Big\{\text{card}~\{a_1, \ldots, a_n\};  \{a_1, \ldots, a_n\}~\text{is a subset of}~\Sigma~\text{shattered by}~\mathcal{G}  \Big\}.$$

One important property of the VC dimension is the fact that $vc(E)\le 1 + \text{dim}~E$ for any finite dimensional vector space of real-valued measurable functions (see e.g Devore and Lorentz \cite{devore}). This is crucial to obtain concrete approximation architectures.

\subsection{Error estimates}
In this section, we present the error estimates of the Longstaff-Schwartz-type algorithm described in previous section. Throughout this section, we fix approximation architectures $\mathcal{H}^k_{N,m}; m=0,\ldots, e(k,T)-1$ which must be chosen according to some a priori information about the smoothness on the continuation values, where $\mathcal{H}^k_{N,0}\subset \mathbb{R}$. For a moment, let us fix them. In what follows, we make use of the standard notation: If $F: W \rightarrow \mathbb{R}$ is a real-valued function defined on a metric space $W$, then we denote $\|F\|_{\infty,W} = \sup_{w\in W}|F(w)|$. When there is no risk o confusion, we just write $\|F\|_\infty$.

We set $\mathbf{A}^k_j:=(\mathcal{A}^k_j, \ldots, \mathcal{A}^k_{e(k,T)})$ and we assume that $\mathbf{A}^k_{jN}$ and $\mathbf{A}^k_{j}$ are independent for each $j=0, \ldots, e(k,T)$ and $N\ge 1$. In the sequel, we employ corresponding lower case notation $\mathbf{a}^k_j$ to denote the deterministic analogue of the random element $\mathbf{A}^k_j$ in such way that $\mathbf{a}^k_j$ is an arbitrary fixed element of $\otimes_{\ell=j}^{e(k,T)}\mathbb{S}^\ell_k$ ($e(k,T)-j+1$ factors). In a similar way, the elements of $\big(\otimes_{\ell=j}^{e(k,T)}\mathbb{S}^\ell_k\big)^N$ ($N$ factors) will be denoted by $\textbf{a}^k_{jN}$ for $j=0, \ldots, e(k,T)$. From $\mathcal{A}^k_{j,i}; i=1, \ldots, N, j=0, \ldots, e(k,T)$, we set $\mathbf{A}^k_{(j,i)}:=(\mathcal{A}^k_{j,i}, \mathcal{A}^k_{(j+1),i}, \ldots, \mathcal{A}^k_{e(k,T),i})$. Of  course, $\mathbf{A}^k_{jN}$ can be naturally identified with $\big(\mathbf{A}^k_{(j,1)}, \ldots, \mathbf{A}^k_{(j,N)}\big)$ for $j=0, \ldots, e(k,T)$ and $N\ge 1$. Moreover, for each $i=1, \ldots, N$, $\mathbf{a}^k_{(j,i)}$ denotes a generic element of $\otimes_{\ell=j}^{e(k,T)}\mathbb{S}^\ell_k$ and, of course, $\mathbf{a}^k_{jN}$ can be identified with $\big(\mathbf{a}^k_{(j,1)}, \ldots, \mathbf{a}^k_{(j,N)}\big)$.

\

In order to prove convergence of the Longstaff-Schwartz algorithm with explicit error estimates, we need to impose the following conditions for a given $k\ge 1$:

\

\noindent \textbf{(H1)} In the context of the Longstaff-Schwartz algorithm, let $N\ge 2$ and we suppose that $\mathcal{H}^k_{N,j}\subset L^2(\mathbb{S}^j_k, \rho^k_j)$ and there exists $\nu_k$ such that $vc\big(\mathcal{H}^k_{N,j}\big)\le \nu_k < +\infty$ for every $j=1, \ldots, e(k,T)-1$ and for every $N\ge 2$.

\

\noindent \textbf{(H2)} There exists $B_k$ such that $\sup\{\|f\|_\infty; f\in \mathcal{H}^k_{N,j}\}\le B_k < +\infty$ for every $j=0, \ldots, e(k,T)-1$ and $N\ge 2$.

\

Assumptions \textbf{(H1-H2)} are well-known in the classical case of the Longstaff-Schwartz algorithm based on Markov chains. See e.g Egloff \cite{egloff}, Zanger \cite{zanger,zanger1} and other references therein.

\begin{remark}
It is important to point out that the approximation architectures in (\textbf{H1-H2}) need not be neither convex nor closed. This relaxation is not surprising and it is known to work in the classical Longstaff-Schwartz algorithm based on Markov chains as demonstrated by Zanger \cite{zanger}.
\end{remark}

\begin{remark}
It is important to stress that we are assuming the existence of a minimizer (\ref{minimizerMC})
and the existence of the projection map $\pi:\mathcal{H}^k_{N,j}\rightarrow L^2(\mathbb{S}^j_k,\rho^k_j)$

$$\pi_{\mathcal{H}^k_{N,j}}(f) = \arg \min_{g\in \mathcal{H}^k_{N,j}}\|g-f\|_{L^2(\mathbb{S}^j_k,\rho^k_j)}$$
for $1\le j\le e(k,T)-1$, $N\ge 2$ and $k\ge 1$. There are some conditions for the existence of both minimizers as discussed in Remark 5.4 by Zanger \cite{zanger}, for instance, compactness of $\mathcal{H}^k_{N,j}$. We observe that we may assume that $\mathbb{S}^j_k; 1\le j\le e(k,T)$ they are compact because all the hitting times $(T^k_n)_{n=1}^{e(k,T)}$ in the control problem have to be restricted to $[0,T]$ and $\mathbb{I}_k$ is obviously compact. In this case, Arzela-Ascoli theorem provides a readable characterization for compact architecture spaces.
\end{remark}


Let us define

\begin{equation}\label{Lk}
L_k = \max\big\{1, \|\mathbb{Z}^k_1\|_\infty, \ldots, \|\mathbb{Z}^k_{e(k,T)}\|_\infty\big\}
\end{equation}
for $k\ge 1$. In the sequel, we set $c_\ell(k) := 2(e(k,T)-\ell+1)\text{log}_2\big( \textbf{e}(e(k,T)-\ell+1)  \big), C_{B_kL_k} := 36(B_k+L_k)^2$ where $\textbf{e}$ is the Euler's number. By using Remark \ref{Markov_A}, the proof of the following lemma follows from the proof of Theorem 5.1 in pages 26-27 given by Zanger \cite{zanger} so we omit the details.

\begin{lemma}\label{fundamentallemma}
Assume that hypotheses \textbf{(H1-H2)} and $L_k < + \infty$ fold true for a given $k\ge 1$. We set $\mathcal{H}^k_{N,0} = [-L_k,L_k]$, Then, for $\alpha >0$ and $\ell=0, \ldots, e(k,T)-1$, we have

$$
\mathbb{P}\Bigg\{ \Big\|\widehat{\mathbf{U}}^k_\ell(\cdot; \mathbf{A}^k_{\ell N}) - \mathbb{E}\big[\mathbf{Z}^k_{\widehat{\tau}^k_{\ell+1}}(\mathbf{A}^k_{\ell+1})|\mathcal{A}^k_\ell=\cdot\big]\Big\|^2_{L^2(\rho^k_\ell)}\ge \alpha
$$
\begin{equation}\label{part}
+4\inf_{g\in\mathcal{H}^k_{N,\ell}}\Big\| g- \mathbb{E}\big[\mathbf{Z}^k_{\widehat{\tau}^k_{\ell+1}}(\mathbf{A}^k_{\ell+1})|\mathcal{A}^k_\ell=\cdot\big] \Big\|^2_{L^2(\rho^k_\ell)}\Bigg\}
\end{equation}
$$\le 6e^4(c_\ell(k)\nu_k+1)^4B^{2\nu_k}_kL_k^{2c_\ell(k)\nu_k}\Big(\frac{512C_{B_kL_k}(B_k+L_k)e}{\alpha}\Big)^{2\nu_k(1+c_\ell(k))}\times\exp\Big( \frac{-N\alpha}{443 C^2_{B_kL_k}} \Big)$$
as long as $N\ge \frac{36 C^2_{B_kL_k}}{\alpha}$. Here, $\widehat{\mathbf{U}}^k_r(\cdot; \mathbf{A}^k_{rN})$ is computed according to (\ref{mini}).
\end{lemma}

\begin{lemma}\label{firstresult}
Assume that hypotheses $\textbf{(H1-H2)}$ hold true. Then, for each $j=0,\ldots, e(k,T)-1$, we have

\begin{eqnarray}
\nonumber\mathbb{E}\|\widehat{\mathbf{U}}^k_j(\cdot; \mathbf{A}^k_{jN}) - \mathbf{U}^k_j\|_{L^2(\rho^k_j)}&\le & (e(k,T)-j)N^{-1/2}\Big(C(B_k+L_k)^2 (\sqrt{\nu_k c_j(k)} log^{1/2}(N) + log^{1/2}(C_{j,k}))\Big)\\
\label{lth1}& &\\
\nonumber&+& 4\sum_{\ell=j}^{e(k,T)-1} \mathbb{E} \Big(\inf_{f\in \mathcal{H}^k_{N,\ell}}\Big\|f - \mathbb{E}[ \mathbf{Z}^k_{\widehat{\tau}^k_{\ell+1}}(\mathbf{A}^k_{\ell+1})|\mathcal{A}^k_\ell]\Big\|_{L^2(\rho^k_\ell)}\Big).
\end{eqnarray}
Here, $C_{j,k} = C(c_j(k) \nu_k+1)^4 B_k^{2\nu_k} L^{2c_j(k)\nu_k}_k \big(C(B_k + L_k)\big)^{3\nu_k(1+c_j(k))}$ for $j=0, \ldots, e(k,T)-1$, where $c_j(k) = 2(e(k,T)-j+1)~log_2\big(\textbf{e} (e(k,T)-j+1)  \big)$, and $C$ is a numerical constant $0 < C <\infty$ which does not depend on $(\nu_k,k,L_k,B_k)$.   Moreover,

\begin{eqnarray}
\nonumber\mathbb{E}|\widehat{V}_0(\mathbf{A}^k_{0N}) - V^k_0|&\le & e(k,T)N^{-1/2}\Big(C(B_k+L_k)^2 (\sqrt{\nu_k c_0(k)} log^{1/2}(N) + log^{1/2}(C_{0,k}))\Big)\\
\label{lth2}& &\\
\nonumber&+&  4\sum_{\ell=1}^{e(k,T)-1} \mathbb{E} \Big(\inf_{f\in \mathcal{H}^k_{N,\ell}}\Big\|f - \mathbb{E}[\mathbf{Z}^k_{\widehat{\tau}^k_{\ell+1}}(\mathbf{A}^k_{\ell+1})|\mathcal{A}^k_\ell]\Big\|_{L^2(\rho^k_\ell)}\Big).
\end{eqnarray}
\end{lemma}
\begin{proof}
The argument is similar to the proof of Theorem 5.1 in Zanger \cite{zanger}. The key point is the error propagation estimate

$$
\Big\|\widehat{\mathbf{U}}^k_j(\cdot; \mathbf{a}^k_{jN}) - \mathbb{E}[\textbf{Z}^k_{\tau^k_{j+1}}(\mathbf{A}^k_{j+1})|\mathcal{A}^k_j=\cdot]\Big\|_{L^2(\rho^k_j)}
$$
\begin{equation}\label{t4.12}
\le 2\sum_{m=j}^{e(k,T)-1}\Big\|\widehat{\mathbf{U}}^k_m(\cdot; \mathbf{a}^k_{mN}) - \mathbb{E}[\textbf{Z}^k_{\widehat{\tau}^k_{j+1}(\textbf{a}^k_{(m+1)N})}(\mathbf{A}^k_{m+1})|\mathcal{A}^k_m=\cdot]\Big\|_{L^2(\rho^k_j)}
\end{equation}
which holds for each $\textbf{a}^k_{jN}\in \big(\otimes_{\ell=j}^{e(k,T)}\mathbb{S}^\ell_k\big)^N$ where $j=0, \ldots, e(k,T)-1$. For the proof of (\ref{t4.12}), we shall use Remark \ref{Markov_A} and then the same arguments given in the proof of Lemma 2.2 in Zanger \cite{zanger1} hold in our context. Now, by invoking Lemma \ref{fundamentallemma} and (\ref{t4.12}), one may proceed in a similar way as in the proof of Theorem 5.1 in Zanger \cite{zanger}, so we prefer to omit the details.
\end{proof}

An almost immediate consequence of Lemma \ref{firstresult} yields the convergence of our Longstaff-Schwartz algorithm for dense approximation architectures.
\begin{theorem}\label{mainTHLS}
Let us fix $k\ge 1$. Let us assume that the hypotheses \textbf{(A1-A2-H1-H2)} hold true and we assume the architecture space $\mathcal{H}^k_{N,j}$ are dense subsets of $L^2(\rho^k_j)$ for each $j=1, \ldots, e(k,T)-1$ and a positive integer $N\ge 2$. Then, for each $j=0, \ldots, e(k,T)-1$, we have

\begin{equation}\label{a.sconv1}
\Big\| \widehat{\mathbf{U}}^k_j(\mathbf{A}^k_{jN}) - \mathbf{U}^k_j\Big\|_{L^2(\rho^k_j)}\rightarrow 0\quad\text{and}~|\widehat{V}^k_0(\mathbf{A}^k_{0N}) - V^k_0|\rightarrow 0
\end{equation}
almost surely as $N\rightarrow \infty$. More importantly, for every $k\ge 1$ sufficiently large

\begin{equation}\label{a.sconv2}
\lim_{N\rightarrow +\infty}|\widehat{V}^k_0(\mathbf{A}^k_{0N}) - S(0)|= 0~a.s.
\end{equation}
\end{theorem}
\begin{proof}
When $\mathcal{H}^k_{N,j}$ are dense subsets of $L^2(\rho^k_j)$ for each $j=1, \ldots, e(k,T)-1$, then we shall use the estimate (\ref{part}) jointly with the fundamental error propagation (\ref{t4.12}) in order to obtain the almost sure convergence (\ref{a.sconv1}) via a standard Borel-Cantelli argument. By Proposition 3.1 in \cite{LEAO_OHASHI2017.3}, we know that

\begin{equation}\label{eqf}
\sup_{\tau\in D^k_{0,T}}\mathbb{E}\big[Z^k(\tau\wedge T^k_{e(k,T)})\big] = \sup_{\tau\in \mathcal{T}_0(\mathbb{F})}\mathbb{E}\big[Z^k(\tau\wedge T^k_{e(k,T)})\big]
\end{equation}
for every $k\ge 1$. Moreover, since $Z^k$ is an imbedded discrete structure for $Z$, then from Theorem 3.2 in \cite{LEAO_OHASHI2017.3}, we have

\begin{equation}\label{eqf1}
\mathbb{E}\sup_{0\le t\le T}|Z^k(t \wedge T^k_{e(k,T)}) - Z(t)|\rightarrow 0
\end{equation}
as $k\rightarrow +\infty$. Then, the proof of (\ref{a.sconv2}) is a simple combination of (\ref{a.sconv1}), (\ref{eqf}), (\ref{eqf1}), triangle inequality and the fact that

\begin{eqnarray*}
V^k_0=\max\{\mathbb{Z}^k_0(\textbf{0}), \mathbf{U}^k_0(\textbf{0})\} = \sup_{\tau\in D^k_{0,T}}\mathbb{E}\big[Z^k(\tau\wedge T^k_{e(k,T)})\big] &=& \sup_{\tau\in \mathcal{T}_0(\mathbb{F})}\mathbb{E}\big[Z^k(\tau\wedge T^k_{e(k,T)})\big]\\
& &\\
&\rightarrow & \sup_{\mathcal{T}_0(\mathbb{F})}\mathbb{E}\big[Z(\tau)\big]
\end{eqnarray*}
as $k\rightarrow\infty$.

\end{proof}

We are now able to state the overall error estimate by separating the stochastic error due to the Monte Carlo procedure and the approximation error due to the use approximation architecture spaces in the regression methodology. This type of estimate is important specially when one has some a priori information on the regularity of the continuation values as described in Section \ref{regCVsection}.

\begin{proposition}\label{maincor}
Let us fix $k\ge 1$. Assume that $L_k < \infty$ and conditions \textbf{(H1-H2)} hold true. For each $j=0,\ldots, e(k,T)-1$, we have

\begin{eqnarray}
\nonumber\mathbb{E}\|\widehat{\mathbf{U}}^k_j(\cdot; \mathbf{A}^k_{jN}) - \mathbf{U}^k_j\|_{L^2(\rho^k_j)}&\le & 6^{e(k,T)-j}\Bigg(\frac{C(B_k+L_k)^2 (\sqrt{\nu_k c_j(k)} log^{1/2}(N) + log^{1/2}(C_{j,k}))}{N^{1/2}}\\
\label{lth3}& &\\
\nonumber&+& \max_{\ell=j,\ldots,e(k,T)-1}\inf_{f\in \mathcal{H}^k_{N,\ell}}\Big\|f - \mathbf{U}^k_\ell\Big\|_{L^2(\rho^k_\ell)}\Bigg).
\end{eqnarray}
Moreover,

\begin{eqnarray}
\nonumber\mathbb{E}|\widehat{V}_0(\mathbf{A}^k_{0N}) - V^k_0|&\le & 6^{e(k,T)}\Bigg(\frac{C(B_k+L_k)^2 (\sqrt{\nu_k c_0(k)} log^{1/2}(N) + log^{1/2}(C_{0,k}))}{N^{1/2}}\\
\label{lth4}& &\\
\nonumber&+& \max_{\ell=1,\ldots,e(k,T)-1}\inf_{f\in \mathcal{H}^k_{N,\ell}}\Big\|f - \mathbf{U}^k_\ell\Big\|_{L^2(\rho^k_\ell)}\Bigg).
\end{eqnarray}
\end{proposition}
\begin{proof}
Keeping in mind Remark \ref{Markov_A}, starting from Lemma \ref{firstresult} and the error propagation (\ref{t4.12}), the proof of (\ref{lth3}) and (\ref{lth4}) follows the same lines of the proof of Theorem 5.6 in Zanger \cite{zanger}, so we omit the details.
\end{proof}

\begin{remark}\label{truncrem}
In Theorem \ref{mainTHLS} and Proposition \ref{maincor}, if the payoffs $\mathbb{Z}^k_n;0\le n\le e(k,T)$ are non-negative, we shall relax the $L^\infty$ boundedness assumption (\ref{Lk}) by boundedness in $L^p$

$$\max\Big\{1, \|\mathbb{Z}^k_1\|^p_{L^p(\rho^k_1)}, \ldots, \|\mathbb{Z}^k_{e(k,T)}\|^p_{L^p(\rho^k_{e(k,T)})}\Big\}< \infty$$
for some $p$ such that $2< p< \infty$. The argument is completely similar to the proof of Prop. 5.2 in Egloff \cite{egloff}. Let us give a brief sketch of the argument for sake of completeness. Let $T_\beta$ be the usual truncation operator and let $\textbf{U}^k_{\beta, j};0\le j\le e(k,T)$ be the continuation value associated with the truncated payoff process $T_\beta (\mathbb{Z}^k_j); 0\le j\le e(k,T)$. From Remark \ref{useful1} (see (\ref{cvalues2})), we have

$$
\|\mathbf{U}^k_j(\mathcal{A}^k_j)  -\textbf{U}^k_{\beta, j}(\mathcal{A}^k_j) \|_{L^p} = \Big\|\mathbb{E}\big[\mathbb{Z}^k_{j+1}(\mathcal{A}^k_{j+1})\vee \mathbf{U}^k_{j+1}(\mathcal{A}^k_{j+1}) - T_\beta\big(\mathbb{Z}^k_{j+1}(\mathcal{A}^k_{j+1})\big)\vee \mathbf{U}^k_{\beta,j+1}(\mathcal{A}^k_{j+1})  |\mathcal{A}^k_j\big]\Big\|_{L^p}
$$
for $j=e(k,T)-1, \ldots, 0$. Then, apply the same steps as in the proof of Prop 5.2 in Egloff \cite{egloff} and the steps outlined in the discussion in pages 15 and 16 in Zanger \cite{zanger}.
\end{remark}

Proposition \ref{maincor} implies that we shall allow the bound $\nu_k(N)$ (as a function of $N$) diverges as $N\rightarrow+\infty$ as long as we are able to control the quantity

\begin{equation}\label{sp1}
\max_{\ell=j,\ldots,e(k,T)-1}\inf_{f\in \mathcal{H}^k_{N,\ell}}\Big\|f - \mathbf{U}^k_\ell\Big\|_{L^2(\rho^k_\ell)}
\end{equation}
and the approximating architecture spaces remain uniformly bounded along the time steps ($0\le j\le e(k,T)-1$) for each $k\ge 1.$ Indeed, from Approximation Theory (see e.g Gyorfi, Kohler, Krzyzak and Walk \cite{gyorfi}), we know that if the continuation values exhibit some degree of Sobolev-type regularity, then we can formulate concrete finite-dimensional architecture spaces to approximate the value function. Let us provide a more precise statement about this.

Let $\mathcal{P}_m(r)$ be the linear space of all polynomials of degree at most $r$ with real coefficients defined on $\mathbb{R}^m$. This it is the space of all linear combinations of terms $x_1^{\alpha_1} \ldots x_m^{\alpha_m}$ with $\sum_{i=1}^m\alpha_i\le r$, $\alpha_i$ being any integer such that $\alpha_i\ge 0$. It is know that (see Lemma 1, chapter 9 in Feinerman and Newman \cite{feiner})

$$\text{dim}(\mathcal{P}_m(r)) = \frac{(r + m)!}{r! m!}.$$
Then, $\text{vc}\big(\mathcal{P}_{m}(r)\big)\le 1 + \text{dim}(\mathcal{P}_m(r))\le 3m^m r^m$. With a slight abuse of notation, when $r>0$, we write $\mathcal{P}_m(r) = \mathcal{P}_m(\lfloor r\rfloor)$ where $\lfloor \cdot\rfloor$ denotes the integer part of a positive number.


We observe that we may assume that driving noise $\mathcal{A}^k_j$ is restricted to a compact subset $K_j$ of $\mathbb{R}^{j(d+1)}$ because all the hitting times $(T^k_n)_{n=1}^{e(k,T)}$ in the control problem have to be restricted to $[0,T]$ and $\mathbb{I}_k$ is obviously compact. Moreover, from the disintegration formula given in Theorem \ref{disintegrationTH}, we observe we can always define a $C^\infty$ extension of $\textbf{b}^k_n\mapsto \nu^k_{n+1}(E|\textbf{b}^k_n)$ from $\mathbb{S}^n_k$ to $\widetilde{\mathbb{S}}^n_k$ (for every $E\in \mathcal{B}(\mathbb{S}_k)$), where $\widetilde{\mathbb{S}}_k: = (0,+\infty)\times B_r(0)$ and $B_r(0)$ is any open ball in $\mathbb{R}^d$ with radius $r>1$. Moreover, we assume that $\mathbb{Z}^k_{e(k,T)}$ can be extended from $\mathbb{S}^{e(k,T)}_k$ to $\widetilde{\mathbb{S}}^{e(k,T)}_k$. This assumption is not strong since the signals of the noise can be easily replaced by real numbers in typical examples found in practice like SDE with random coefficients. See Section \ref{regCVsection} for a concrete example.

Recall the generic form of the reward process is $Z = F(X)$ where $X$ is some state process which is a functional of the $d$-dimensional Brownian motion. Of course, we are assuming that $Z$ admits an imbedded discrete structure $Z^k$. Since $Z$ is $\mathbb{F}$-adapted $F:\Lambda_T\rightarrow\mathbb{R}$ has to be non-anticipative (see (\ref{nonantdef})) where $\Lambda_T$ is the space of the stopped paths as described in (\ref{stopping_times}). In the sequel, $W^{1}\big(L^\infty(K_j)\big)$ is the usual Sobolev space equipped with the standard norm $\|\cdot\|_{\infty,1,K_j}$. Moreover, we denote $\|f\|_{\infty,K_j}:=\sup_{x\in K_j}|f(x)|$.

\begin{corollary}\label{contVABS}
Assume \textbf{(H1-H2)} hold true, the reward functional $F:\Lambda_T\rightarrow\mathbb{R}$ is bounded, $L_k < +\infty$ for $k\ge 1$ and the continuation values $\textbf{U}^k_j$ are in the Sobolev spaces $W^{1}\big(L^\infty(K_j)\big)$ for $j=1,\ldots, e(k,T)-1$. Let us define the sequences of approximation architectures

$$
\mathcal{H}^k_{N,j}=\Big\{ p\in \mathcal{P}_{j(d+1)}\big(N^{1/j(d+1) + 2}\big) ; \|p\|_{\infty,K_j}\le 2 \|\textbf{U}^k_j\|_{\infty,1,K_j}\Big\}, \mathcal{H}^k_{N,0} = [-L_k,L_k].
$$
Then, for $j=0,\ldots, e(k,T)-1$, we have

$$
\mathbb{E}\|\widehat{\mathbf{U}}^k_j(\cdot; \mathbf{A}^k_{jN}) - \mathbf{U}^k_j\|_{L^2(\rho^k_j)}= O\big(\text{log}(N)N^{\frac{-2}{2+e(k,T)-1}}\big).
$$
In particular,
$$\mathbb{E}|\widehat{V}_0(\mathbf{A}^k_{0N}) - V^k_0|= O\big(\text{log}(N)N^{\frac{-2}{2+e(k,T)-1}}\big).
$$
\end{corollary}
\begin{proof}
The argument of the proof is similar to Corollary 5.5 in Egloff \cite{egloff} due to Remark \ref{Markov_A}. So we omit the details.
\end{proof}

\section{Regularity properties of continuation values for path-dependent SDEs}\label{regCVsection}
In this section, we provide a detailed study on the analytical properties of continuation values $\textbf{U}^k_j; 0\le j\le e(k,T)-1$ arising from an optimal stopping problem where the reward process $Z = F(X)$ is a functional of a path-dependent SDE $X$. The degree of regularity of $\textbf{U}^k_j$ is crucial to obtain overall error estimates and concrete approximation spaces as described in Corollary \ref{contVABS}. For simplicity of exposition, the dimension of the SDE and the Brownian motion driving noise will be taken equal to one.

Let us introduce some functional spaces which will play an important role for us. Let $D([0,t];\mathbb{R})$ be the linear space of $\mathbb{R}$-valued c\`adl\`ag paths on $[0,t]$ and

$$\omega_t: = \omega(t\wedge \cdot); \omega \in D([0,T];\mathbb{R}).$$
This notation is naturally extended to processes. We set

\begin{equation}\label{stoppedpaths}
\Lambda_T:=\{(t,\omega_t); t\in [0,T]; \omega\in D([0,T];\mathbb{R})\}
\end{equation}
as the space of stopped paths. In the sequel, a functional will be just a mapping $G:[0,T]\times D([0,T];\mathbb{R}^n)\rightarrow \mathbb{R}; (t,\omega)\mapsto G(t,\omega)$. We endow $\Lambda_T$ with the metric

$$d_{\beta}((t,\omega); (t',\omega')): = \sup_{0\le u\le T}|\omega(u\wedge t) - \omega'(u\wedge t')| + |t-t'|^{\beta};0< \beta\le 1$$
so that $(\Lambda_T,d_\beta)$ is a complete metric space equipped with the Borel sigma-algebra. We say that $G$ is a \textit{non-anticipative} functional if it is a Borel mapping and

\begin{equation}\label{nonantdef}
G(t,\omega) = G(t,\omega_t); (t,\omega)\in[0,T]\times D([0,T];\mathbb{R}).
\end{equation}
In this case, a non-anticipative functional can be seen as a measurable mapping $G:\Lambda_T\rightarrow \mathbb{R}; (t,\omega_t)\mapsto G(t,\omega) = G(t,\omega_t)$ for $(t,\omega_t)\in \Lambda_T$.

The underlying state process is the following $n$-dimensional SDE

\begin{equation}\label{pdsdeBM}
dX(t) = \alpha(t,X_t)dt + \sigma(t,X_t)dB(t); 0\le t\le T,
\end{equation}
with a given initial condition $X(0)=x\in \mathbb{R}$. The coefficients of the SDE will satisfy the following regularity conditions:

\

\noindent \textbf{Assumption I}: The non-anticipative mappings $\alpha: \Lambda_T\rightarrow \mathbb{R}$ and $\sigma:\Lambda_T\rightarrow \mathbb{R}$ are Lipschitz continuous, i.e., there exists a constant $K_{Lip}>0$ such that

$$|\alpha(t,\omega_t) - \alpha(t',\omega'_{t'})| + |\sigma(t,\omega_t) - \sigma(t',\omega'_{t'})|\le K_{Lip}d_{1/2} \big((t,\omega); (t',\omega')\big)$$
for every $t,t'\in [0,T]$ and $\omega,\omega'\in D([0,T];\mathbb{R})$.

\

One can easily check by routine arguments that the SDE (\ref{pdsdeBM}) admits a strong solution such that

$$\mathbb{E}\sup_{0\le t\le T}|X(t)|^{2p}\le C(1+|x_0|^{2p})\exp(CT)$$
where $X(0)=x_0$, $C$ is a constant depending on $T>0,p\ge 1$, $K_{Lip}$.

The reward process of the optimal stopping problem is given by $Z(t) = F(t,X_t)$, where $F$ is a non-anticipative functional $F$. We will assume the following hypothesis on this functional

\

\textbf{Assumption II} The reward process is given by $Z(t) = F(t,X_t)$ where $X$ is the path-dependent SDE (\ref{pdsdeBM}) driven by a Brownian motion. The non-anticipative functional $F:\Lambda_T\rightarrow\mathbb{R}$ has linear growth: There exists a constant $C$ such that

$$|F(t,\omega_t)|\le C (1 + \sup_{0\le t\le T}|\omega(t)|)$$
for every $\omega\in D([0,T];\mathbb{R})$. Moreover, $F:\Lambda_T\rightarrow\mathbb{R}$ is continuous, where $\Lambda_T$ is equipped with the metric $d_\beta$.



\

One can readily see that under Assumption II, the natural candidate for an imbedded discrete structure w.r.t. $Z$ is given by

\begin{equation}\label{REWGAS}
Z^k(t) = \sum_{n=0}^{\infty}F_{T^k_n}\big(X^k_{T^k_n}\big)\mathds{1}_{T^k_n\le t < T^k_{n+1}}; 0\le t\le T
\end{equation}
where $X^k$ is an imbedded discrete structure for the path-dependent SDE $X$ given by (\ref{pdsdeBM}). There exists a natural choice of an imbedded discrete structure for $X$ in terms of an Euler-Maruyama scheme and this was studied in detailed in the works \cite{LEAO_OHASHI2017.2,LEAO_OHASHI2017.3}. For sake of completeness, we provide the construction of the Euler scheme in our setup written on the random partition $\{T^k_n; k,n\ge 1\}$. We start $X^k(0):=x$ and we proceed by induction

\begin{eqnarray*}
X^{k}(T^k_{m})&:=&X^{k}(T^k_{m-1}) + \alpha\big(T^k_{m-1},X^{k}_{T^k_{m-1}}\big)\Delta T^k_{m}\\
& &\\
&+& \sigma(T^k_{m-1},X^{k}_{T^k_{m-1}}\big)\Delta A^{k,1}(T^k_{m})
\end{eqnarray*}
for $1\le m\ge 1$. We then define $X^k(t):=\sum_{\ell=0}^{\infty}X^{k}(T^k_\ell) 1\!\!1_{\{T^k_\ell\le t < T^k_{\ell+1}\}}; 0\le t\le T $.

Under Assumptions I-II, \cite{LEAO_OHASHI2017.2,LEAO_OHASHI2017.3} show that $Z^k$ given by (\ref{REWGAS}) is an imbedded discrete structure for the reward process $Z = F(X)$ where $X$ is the path-dependent SDE (\ref{pdsdeBM}) (see Prop. 5.1 in \cite{LEAO_OHASHI2017.3}). In order to establish regularity of the continuation values $\textbf{U}^k_j; 0\le j\le e(k,T)-1$ associated with $Z^k$, we need to work pathwise.

\subsection{Pathwise description of the path-dependent SDE:}
In order to investigate the regularity of the continuation values, we need a pathwise representation of the Euler-Maruyama scheme. Initially, we define $h^k_0:=x$ and then we proceed by induction,

\begin{eqnarray*}
h^{k}_{m}(\textbf{b}^{k}_m)&:=&h^{k}_{m-1}(\textbf{b}^{k}_{m-1}) + \alpha\big(t^k_{m-1},\bar{\gamma}^k_{m-1}(\textbf{b}^{k}_{m-1})\big)s^k_{m}\\
& &\\
&+& \sigma \Big(t^k_{m-1},\bar{\gamma}^k_{m-1}(\textbf{b}^{k}_{m-1})\Big)
\epsilon_k\tilde{i}^{k}_{m}; ~\textbf{b}^{k}_m\in \mathbb{S}^m_k,
\end{eqnarray*}
for $m\ge 1$, where

$$\bar{\gamma}^k_{m-1}(\textbf{b}^{k}_{m-1})(t):= \sum_{\ell=0}^{m-1}h^k_{\ell}(\textbf{b}^{k}_{\ell}) 1\!\!1_{\{t^k_\ell\le t < t^k_{\ell+1}\}}; 0\le t\le T,$$
and we recall $t^k_n$ is defined in (\ref{tknfunction}). We then define
$$ \bar{\gamma}^k(\textbf{b}^{k}_{\infty})(t):= \sum_{n=0}^\infty h^k_{n}(\textbf{b}^{k}_{n})1\!\!1_{\{t^k_n \le t< t^k_{n+1}\}}$$
for $\textbf{b}^{k}_{\infty}\in \mathbb{S}^{\infty}_k$ and
$$\gamma^k_{j}(\textbf{b}^{k}_{j})(t):= \bar{\gamma}^k(\textbf{b}^{k}_{\infty})(t\wedge t^k_{j}); 0\le t\le T, j=0\ldots, e(k,T).$$
By the very definition,

$$\gamma^k_{\infty}\Big(\mathcal{A}^k_\infty(\omega)\Big)(t) = X^{k}(t,\omega)$$
for a.a $\omega$ and for each $t\in [0,T]$, where $\mathcal{A}^k_{\infty}:=\{\mathcal{A}^k_n; n\ge 0\}$. Hence,

$$\mathbb{Z}^k_{j}(\mathcal{A}^k_j) = Z^k(T^k_j\wedge T) = \sum_{n=0}^\infty F(T^k_n, X^k_{T^k_n})1\!\!1_{\{T^k_n\le T^k_j \wedge T < T^k_{n+1}\}}$$
for $j=0,\ldots, e(k,T)$, where

\begin{equation}\label{Zpathwise}
\mathbb{Z}^k_{j}(\textbf{b}^k_j) = \sum_{n=0}^\infty F(t^k_n,\gamma^k_n(\textbf{b}^k_n))1\!\!1_{\{t^k_n\le t^k_j \wedge T < t^k_{n+1}\}};j=e(k,T), \ldots, 0.
\end{equation}

\subsection{Regularity properties of the continuation values:} Let us now analyse the analytical properties of the continuation values

$$
\textbf{b}^k_j\mapsto \mathbf{U}^k_j(\textbf{b}^k_j)=\mathbb{E}\Big[\mathbb{Z}^k_{\tau^{k}_{j+1}}(\mathcal{A}^k_{\tau^{k}_{j+1}})\big|\mathcal{A}^k_j=\textbf{b}^k_j\Big]; 0\le j\le e(k,T)-1.
$$

From Remark \ref{useful1},

$$
\mathbf{U}^k_j(\textbf{b}^k_j) = \int_{\mathbb{S}_k}\mathbb{Z}^k_{j+1}(\textbf{b}^k_{j},s^k_{j+1},\tilde{i}^k_{j+1} )\vee \mathbf{U}^k_{j+1}(\textbf{b}^k_j,s^k_{j+1},\tilde{i}^k_{j+1})\mathbb{P}\big[(\Delta T^k_{j+1},\eta^k_{j+1})\in (ds^k_{j+1},d\tilde{i}^k_{j+1})|\mathcal{A}^k_j=\textbf{b}^k_j\big]
$$
for every $\textbf{b}^k_j\in \mathbb{S}^j_k$.

In order to get Sobolev-type regularity for the continuation values, we require more regularity from reward functional $F$ and the coefficients of the SDE (\ref{pdsdeBM}):

\

\noindent \textbf{Assumption III}: The reward functional $F$ and the coefficients $\alpha$ and $\sigma$ of the SDE are bounded. Moreover, there exist constants $\bar{K}_{Lip}$ and $|F|$ such that

$$|\alpha(t,\omega_t) - \alpha(t',\omega'_{t'})| + |\sigma(t,\omega_t) - \sigma(t',\omega'_{t'})|\le \bar{K}_{Lip}\big\{|t-t'| + \rho(\omega_t,\omega'_{t'})\big\}$$
and

$$|F(t,\omega_t) - F(t',\omega'_{t'})| \le |F| \big\{|t-t'| + \rho(\omega_t,\omega'_{t'})\big\};~t,t'\in [0,T],\omega,\omega'\in D([0,T];\mathbb{R}),$$
where $\rho$ is the standard metric on $D([0,T];\mathbb{R})$ generating the Skorohod topology given by

$$\rho(f,g):=\inf_{\lambda\in \mathcal{K}}\Big(\|\lambda-I\|_\infty \vee \|f-g(\lambda)\|_\infty\Big)~;~f,g\in D([0,T];\mathbb{R})$$
where $\mathcal{K}$ is the set of all strictly increasing functions from $[0,T]$ onto $[0,T]$.


We are now able to present the main result of this section. At first, for a Borelian set $E\times F$ of $\mathbb{S}_k$, we shall apply the strong Markov property to get

$$\nu^k_{n+1}(E\times F|\mathbf{b}^k_n) =\mathbb{P}\{(\Delta T^k_{n+1},\eta^k_{n+1})\in E\times F\} = \int_E f_k(x)dx \times \mathbb{P}\{\eta^k_{1}\in F\}$$
where $f_k$ is the density (see e.g Burq and Jones \cite{Burq_Jones2008}) of the i.i.d sequence $\Delta T^k_n; n\ge 1$ and $\eta^k_1$ is a $1/2$-Bernoulli variable which is independent from $\Delta T^k_1$.

We observe that for each $(s^k_1,\ldots, s^k_n)\in (0,+\infty)^n$, we can naturally define $\textbf{U}^k_n(s^k_1,x_1,\ldots, s^k_n,x_n)$ for $-r < x_i < r; i=1, \ldots, n$. This is possible because the sign of the jumps $(\tilde{i}^k_1, \ldots, \tilde{i}^k_n)$ enter linearly in $h^k_\ell; 0\le n\le e(k,T)$. We then write

$$\widetilde{\mathbb{S}}^n_k:=\big((0,+\infty)\times B_r(0)\big) \times \ldots \times \big((0,+\infty)\times B_r(0)\big)\quad (n\text{-fold cartesian product)},$$
and with a slight abuse of notation, a generic element of $\widetilde{\mathbb{S}}^n_k$ will still be denoted by $\textbf{b}^k_n=(s^k_1, x_1, \ldots, s^k_n,x_n)$.

\begin{theorem}\label{regPSDE}
If Assumption \textbf{III} holds true, then for each $n=e(k,T)-1, \ldots, 0,$

$$\mathbf{b}^k_n\mapsto \textbf{U}^k_{n}(\mathbf{b}^k_n)$$
is globally Lipschitz continuous from $\widetilde{\mathbb{S}}^n_k$ to $\mathbb{R}$.
\end{theorem}
\begin{proof}
In the sequel, $C$ is a constant which may defer from line to line. Due to Remark \ref{useful1} (see (\ref{cvalues2})), the first step $m=e(k,T)-1$ is

$$
\mathbf{U}^k_m(\textbf{b}^k_m) = \int_{\mathbb{S}_k}\mathbb{Z}^k_{m+1}(\textbf{b}^k_{m},s^k_{m+1},x_{m+1}) \mathbb{P}_{(\Delta T^k_{m+1},\eta^k_{m+1})}(ds^k_{m+1},dx_{m+1})
$$
for $\textbf{b}^k_m\in \widetilde{\mathbb{S}}^m_k$ where

$$
\mathbb{Z}^k_{m+1}(\textbf{b}^k_{m+1})=\left\{
\begin{array}{rl}
F(t^k_{m+1},\gamma^k_{m+1}(\textbf{b}^k_{m+1})); & \hbox{if} \ s^k_{m+1}\le T-t^k_m  \\
F(t^k_{n_T},\gamma^k_{n_T}(\textbf{b}^k_{n_T}));& \hbox{if} \  T-t^k_m < s^k_{m+1}
\end{array}
\right.
$$
and $n_T$ is the integer ($n_T < m+1$) which realizes $t^k_{n_T}\le T < t^k_{n_T+1}$ whenever $T-t^k_m < s^k_{m+1}$.

Recall that

$$\gamma_{m+1}(\textbf{b}^k_{m+1}) = \sum_{n=0}^\infty h^k_n(\textbf{b}^k_n)1\!\!1_{\{t^k_n\le t < t^k_{n+1}\}} + h^k_{m+1}(\textbf{b}^k_{m+1})1\!\!1_{\{t\ge t^k_{m+1}\}}; 0\le t\le T,$$
in case $t^k_{m+1}\le T$ and

$$\gamma_{n_T}(\textbf{b}^k_{n_T}) = \sum_{j=0}^{n_T-1} h^k_j(\textbf{b}^k_j)1\!\!1_{\{t^k_j\le t < t^k_{j+1}\}} + h^k_{n_T}(\textbf{b}^k_{n_T})1\!\!1_{\{t\ge t^k_{n_T}\}}; 0\le t\le T,$$
in case $t^k_{m+1}> T$. We shall write

\begin{eqnarray*}
h^k_{m+1}(\textbf{b}^k_{m+1}) &=& h^k_{m+1}(\textbf{b}^k_m, s^k_{m+1},\tilde{i}^k_{m+1}) = h^{k}_{m}(\textbf{b}^{k}_{m}) + \alpha^i\big(t^k_{m},\bar{\gamma}^k_{m}(\textbf{b}^{k}_{m})\big)s^k_{m+1}\\
& &\\
&+&  \sigma\Big(t^k_m,\bar{\gamma}^k_{m}(\textbf{b}^{k}_{m})\Big)
\epsilon_k x_{m+1}; ~\textbf{b}^{k}_{m+1}\in \widetilde{\mathbb{S}}^{m+1}_k.
\end{eqnarray*}

Since translation and finite sum are smooth operations, it is sufficient to check for $m=2$. By the very definition,

$$\int_{\mathbb{S}_k}\mathbb{Z}^k_3(\textbf{b}^k_2,s^k_3,x_3)\mathbb{P}_{(\Delta T^k_{3},\eta^k_{3})}(ds^k_{3},d\tilde{i}^k_3) = \frac{1}{2}\int_0^{+\infty}\mathbb{Z}^k_3(\textbf{b}^k_2,s^k_3,1)f_k(s^k_3)ds^k_3 +\frac{1}{2}\int_0^{+\infty}\mathbb{Z}^k_3(\textbf{b}^k_2,s^k_3,-1)f_k(s^k_3)ds^k_3 $$
where $f_{k}$ is the density of $\Delta T^k_3$. To alleviate notation, we set $\theta_k(ds^k_3) = f_k(s^k_3)(ds^k_3)$. Let us take $\textbf{b}^k_2,\bar{\textbf{b}}^k_2\in \widetilde{\mathbb{S}}^2_k$ and notice that


$$\int_0^{+\infty}\big(\mathbb{Z}^k_3(\textbf{b}^k_2,s^k_3,1) -\mathbb{Z}^k_3(\bar{\textbf{b}}^k_2,s^k_3,1)\big) \theta_k(ds^k_3)  =
\int_0^{(T -t^k_2) \wedge (T -\bar{t}^k_2) }\big(\mathbb{Z}^k_3(\textbf{b}^k_2,s^k_3,1) - \mathbb{Z}^k_3(\bar{\textbf{b}}^k_2,s^k_3,1)\big)\theta_k(ds^k_3)$$
$$+ \int_{(T -t^k_2) \wedge (T -\bar{t}^k_2)}^{+\infty}\big(\mathbb{Z}^k_3(\textbf{b}^k_2,s^k_3,1) - \mathbb{Z}^k_3(\bar{\textbf{b}}^k_2,s^k_3,1)\big)\theta_k(ds^k_3)$$
We split the proof into two parts:

\

\textbf{PART 1:} Assumption III yields
\begin{small}
$$\Bigg|\int_0^{(T -t^k_2) \wedge (T -\bar{t}^k_2) }\big(\mathbb{Z}^k_3(\textbf{b}^k_2,s^k_3,1) - \mathbb{Z}^k_3(\bar{\textbf{b}}^k_2,s^k_3,1)\big)\theta_k(ds^k_3)\Bigg|$$
$$ = \Bigg|\int_0^{(T -t^k_2) \wedge (T -\bar{t}^k_2) }\Big[F\big(s^k_3+t^k_2,\gamma^k_3(\textbf{b}^k_2,s^k_3,1)\big) - F\big(s^k_3+\bar{t}^k_2,\gamma^k_3(\bar{\textbf{b}}^k_2,s^k_3,1)\big)\Big]\theta_k(ds^k_3)\Bigg|$$
$$\le \int_0^{(T -t^k_2) \wedge (T -\bar{t}^k_2) }\Big|F\big(s^k_3+t^k_2,\gamma^k_3(\textbf{b}^k_2,s^k_3,1)\big) - F\big(s^k_3+\bar{t}^k_2,\gamma^k_3(\bar{\textbf{b}}^k_2,s^k_3,1)\big)\Big|\theta_k(ds^k_3)$$
\begin{equation}\label{rr1}
\le |F| \int_0^{(T -t^k_2) \wedge (T -\bar{t}^k_2) } \Big[|t^k_2-\bar{t}^k_2| + \rho\big(\gamma^k_3(\textbf{b}^k_2,s^k_3,1),\gamma^k_3(\textbf{b}^k_2,s^k_3,1) \big)\Big]\theta_k(ds^k_3)
\end{equation}
\end{small}
where

$$\gamma^k_3(\textbf{b}^k_2,s^k_3,1)(t) = \sum_{j=0}^{2}h^k_j(\textbf{b}^k_j) 1\!\!1_{\{t^k_j\le t < t^k_{j+1}\}} +  h^k_3(\textbf{b}^{k}_2, s^k_3,1)1\!\!1_{\{ s^k_3+t^k_2\le t\}} ; 0\le t\le T$$
and
$$\gamma^k_3(\bar{\textbf{b}}^k_2,s^k_3,1)(t) = \sum_{j=0}^{2}h^k_j(\bar{\textbf{b}}^k_j) 1\!\!1_{\{\bar{t}^k_j\le t < \bar{t}^k_{j+1}\}} +  h^k_3(\bar{\textbf{b}}^{k}_2, s^k_3,1)1\!\!1_{\{ s^k_3+\bar{t}^k_2\le t\}} ; 0\le t\le T.$$

The Skorohod topology yields (see e.g Example 15.11 in He, Wang and Yan \cite{he})

\begin{small}
$$\rho\big(\gamma^k_3(\textbf{b}^k_2,s^k_3,1),\gamma^k_3(\textbf{b}^k_2,s^k_3,1) \big)\le C \Big(\max_{1\le \ell\le 2}|s^k_\ell - \bar{s}^k_\ell|\vee \max_{1\le \ell\le 2}|h^k_\ell(\bar{\textbf{b}}^{k}_\ell) - h^k_\ell(\textbf{b}^{k}_\ell)|\vee |h^k_3(\bar{\textbf{b}}^{k}_2,s^k_3,1) - h^k_3(\textbf{b}^{k}_2,s^k_3,1)|\Big) $$
\end{small}
for a constant $C$ which only depends on $T$. Recall that
\begin{small}
$$h^{k}_{1}(\textbf{b}^{k}_1) = x + \alpha(0,x)s^k_1 + \sigma(0, x)\epsilon_kx_1,\quad  h^{k}_{1}(\bar{\textbf{b}}^{k}_1) = x + \alpha(0,x)\bar{s}^k_1 + \sigma(0, x)\epsilon_k\overline{x}_1$$

$$h^{k}_2(\textbf{b}^{k}_2)=h^{k}_1(\textbf{b}^{k}_1) + \alpha\big(t^k_1,\gamma^k_1(\textbf{b}^{k}_1)\big)s^k_2+\sigma\Big(t^k_1,\gamma^k_{1}(\textbf{b}^{k}_1)\Big)\epsilon_kx_2,\quad h^{k}_2(\bar{\textbf{b}}^{k}_2)=h^{k}_1(\bar{\textbf{b}}^{k}_1) + \alpha\big(\bar{t}^k_1,\gamma^k_1(\bar{\textbf{b}}^{k}_1)\big)\bar{s}^k_2+\sigma\Big(\bar{t}^k_1,\gamma^k_{1}(\bar{\textbf{b}}^{k}_1)\Big)\epsilon_k\overline{x}_2 $$

$$\gamma^k_\ell(\textbf{b}^{k}_\ell)(t)=\sum_{r=0}^{\ell-1}h^k_r(\textbf{b}^k_r)1\!\!1_{\{t^k_r\le t < t^k_{r+1}\}}+ h^k_\ell(\textbf{b}^{k}_\ell)1\!\!1_{\{t^k_\ell\le t\}},\quad \gamma^k_\ell(\bar{\textbf{b}}^{k}_\ell)(t)=\sum_{r=0}^{\ell-1}h^k_r(\bar{\textbf{b}}^k_r)1\!\!1_{\{\bar{t}^k_r\le t < \bar{t}^k_{r+1}\}}+ h^k_\ell(\bar{\textbf{b}}^{k}_\ell)1\!\!1_{\{\bar{t}^k_\ell\le t\}}
$$
\end{small}

and
\begin{small}
$$
h^{k}_3(\textbf{b}^{k}_2,s^k_3,1)=h^{k}_2(\textbf{b}^{k}_2) + \alpha\big(t^k_2,\gamma^k_2(\textbf{b}^{k}_2)\big)s^k_3+\sigma\Big(t^k_2,\gamma^k_{2}(\textbf{b}^{k}_2)\Big)\epsilon_k,
$$
$$
h^{k}_3(\bar{\textbf{b}}^{k}_2,s^k_3,1)=h^{k}_2(\bar{\textbf{b}}^{k}_2) + \alpha\big(\bar{t}^k_2,\gamma^k_2(\bar{\textbf{b}}^{k}_2)\big)s^k_3+\sigma\Big(\bar{t}^k_2,\gamma^k_{2}(\bar{\textbf{b}}^{k}_2)\Big)\epsilon_k.
$$
\end{small}
Then,
\begin{small}
$$|h^k_1(\textbf{b}^k_1) - h^k_1(\bar{\text{b}}^k_1)|\le |\alpha(0,x)| |s^k_1-\bar{s}^k_1| + |\sigma(0,x)| |x_1 - \overline{x}_1|$$
and
$$|h^k_2(\textbf{b}^k_2) - h^k_2(\bar{\text{b}}^k_2)|\le |h^k_1(\textbf{b}^k_1) - h^k_1(\bar{\text{b}}^k_1)|+ s^k_2 |\alpha(t^k_1,\gamma^k_1(\textbf{b}^k_1)) -\alpha(\bar{t}^k_1,\gamma^k_1(\bar{\textbf{b}}^k_1))| + |\alpha(\bar{t}^k_1,\gamma^k_1(\bar{\textbf{b}}^k_1))||s^k_2-\bar{s}^k_2|$$
$$+\epsilon_k |x_2| |  \sigma(t^k_1,\gamma^k_1(\textbf{b}^k_1)) -\sigma(\bar{t}^k_1,\gamma^k_1(\bar{\textbf{b}}^k_1))| + |\sigma(\bar{t}^k_1,\gamma^k_1(\bar{\textbf{b}}^k_1))||x_2-\overline{x}_2|.$$
\end{small}
Similarly,

$$|h^k_3(\textbf{b}^k_2,s^k_3,1) - h^k_3(\bar{\text{b}}^k_2, s^k_3,1)|\le |h^k_2(\textbf{b}^k_2) - h^k_2(\bar{\text{b}}^k_2)|+ s^k_3 |\alpha(t^k_2,\gamma^k_2(\textbf{b}^k_2)) -\alpha(\bar{t}^k_2,\gamma^k_2(\bar{\textbf{b}}^k_2))|$$
$$+\epsilon_k|  \sigma(t^k_2,\gamma^k_2(\textbf{b}^k_2)) -\sigma(\bar{t}^k_2,\gamma^k_2(\bar{\textbf{b}}^k_2))|.$$

Then, by using Assumption III, we get

\begin{equation}\label{rr2}
\rho\big(\gamma^k_3(\textbf{b}^k_2,s^k_3,1),\gamma^k_3(\textbf{b}^k_2,s^k_3,1) \big)\le C \Big(\max_{1\le\ell\le 2}|s^k_\ell - \bar{s}^k_\ell|\vee \max_{1\le \ell\le 2}|x_\ell-\overline{x}_\ell|\Big)
\end{equation}
for a constant $C$ which depends on $T,\epsilon_k,m, K_{Lip},\alpha(0,x)$ and $\sigma(0,x)$. By plugging in (\ref{rr2}) into (\ref{rr1}) and using the fact that $\theta_k$ is a probability measure, we conclude that
$$\Bigg|\int_0^{(T -t^k_2) \wedge (T -\bar{t}^k_2) }\big(\mathbb{Z}^k_3(\textbf{b}^k_2,s^k_3,1) - \mathbb{Z}^k_3(\bar{\textbf{b}}^k_2,s^k_3,1)\big)\theta_k(ds^k_3)\Bigg|\le C \|\textbf{b}^k_2 - \bar{\textbf{b}}^k_2\|_{\mathbb{R}^4}.
$$

\

\textbf{PART 2:}
\begin{small}
\begin{equation}\label{rr3}
\int_{(T -t^k_2) \wedge (T -\bar{t}^k_2)}^{+\infty}\big(\mathbb{Z}^k_3(\textbf{b}^k_2,s^k_3,1) - \mathbb{Z}^k_3(\bar{\textbf{b}}^k_2,s^k_3,1)\big)\theta_k(ds^k_3) = \int_{(T -t^k_2) \wedge (T -\bar{t}^k_2)}^{(T -t^k_2) \vee (T -\bar{t}^k_2)}\big(\mathbb{Z}^k_3(\textbf{b}^k_2,s^k_3,1) - \mathbb{Z}^k_3(\bar{\textbf{b}}^k_2,s^k_3,1)\big)\theta_k(ds^k_3)
\end{equation}
$$+ \int_{(T -t^k_2) \vee (T -\bar{t}^k_2)}^{+\infty}\big(\mathbb{Z}^k_3(\textbf{b}^k_2,s^k_3,1) - \mathbb{Z}^k_3(\bar{\textbf{b}}^k_2,s^k_3,1)\big)\theta_k(ds^k_3).$$
\end{small}
Due to the boundedness assumption on $F$, the first term in the right hand side of (\ref{rr3}) is bounded by $\max_{1\le \ell\le 2}|s^k_\ell - \bar{s}^k_\ell|$. We notice there exists a constant $C$ such that

$$\int_{(T -t^k_2) \vee (T -\bar{t}^k_2)}^{+\infty}\big(\mathbb{Z}^k_3(\textbf{b}^k_2,s^k_3,1) - \mathbb{Z}^k_3(\bar{\textbf{b}}^k_2,s^k_3,1)\big)\theta_k(ds^k_3) \le C \|\textbf{b}^k_2 - \bar{\textbf{b}}^k_2\|_{\mathbb{R}^4}
$$
whenever $n_T(s^k_1,s^k_2) = n_T(\bar{s}^k_1,\bar{s}^k_2)$. But $n_T(s^k_1,s^k_2) = n_T(\bar{s}^k_1,\bar{s}^k_2)$ holds true as long as $\max_{1\le \ell\le 2}|s^k_\ell-\bar{s}^k_\ell|$ is small.

Summing up PART1 and PART2, we then infer the existence of a constant $C$ such that

$$
\Bigg|\int_0^{+\infty}\big(\mathbb{Z}^k_3(\textbf{b}^k_2,s^k_3,1) -\mathbb{Z}^k_3(\bar{\textbf{b}}^k_2,s^k_3,1)\big) \theta_k(ds^k_3) \Bigg|\le C\|\textbf{b}^k_2 - \bar{\textbf{b}}^k_2\|_{\mathbb{R}^4}
$$
for every $\textbf{b}^k_2,\bar{\textbf{b}}^k_2\in \widetilde{\mathbb{S}}^2$. Similarly,
$$\Bigg|\int_0^{+\infty}\big(\mathbb{Z}^k_3(\textbf{b}^k_2,s^k_3,-1) -\mathbb{Z}^k_3(\bar{\textbf{b}}^k_2,s^k_3,-1)\big) \theta_k(ds^k_3) \Bigg|\le C\|\textbf{b}^k_2 - \bar{\textbf{b}}^k_2\|_{\mathbb{R}^4}$$
for every $\textbf{b}^k_2,\bar{\textbf{b}}^k_2\in \widetilde{\mathbb{S}}^2_k$. This shows that $\textbf{b}^k_2\mapsto \textbf{U}^k_2(\textbf{b}^k_2)$ is Lipschitz.

\

\textbf{PART 3:} By Remark \ref{useful1} (see (\ref{cvalues2})), we have

$$\mathbf{U}^k_{1}(\textbf{b}^k_1) = \int_{\mathbb{S}_k}\mathbb{Z}^k_{2}(\textbf{b}^k_{1},s^k_{2},x_2 )\vee \mathbf{U}^k_{2}(\textbf{b}^k_1,s^k_{2},x_2)\mathbb{P}_{(\Delta T^k_{2},\eta^k_{2})}(ds^k_{2},dx_2); \textbf{b}^k_1\in \widetilde{\mathbb{S}}_k
$$
By using the elementary inequality $|a\vee b - c\vee d|\le |a\vee b - a\vee d| + |a\vee d- c\vee d|\le |b-d| + |a-c|; a,b,c,d\in \mathbb{R}$ and the previous step, we have
\begin{small}
\begin{eqnarray*}
|\mathbf{U}^k_{1}(\textbf{b}^k_1) - \mathbf{U}^k_{1}(\bar{\textbf{b}}^k_1)|&\le&\int_{\mathbb{S}_k}\Big|\mathbb{Z}^k_{2}(\textbf{b}^k_{1},s^k_{2},x_2 )\vee \mathbf{U}^k_{2}(\textbf{b}^k_1,s^k_{2},x_2) - \mathbb{Z}^k_{2}(\bar{\textbf{b}}^k_{1},s^k_{2},x_2 )\vee \mathbf{U}^k_{2}(\bar{\textbf{b}}^k_1,s^k_{2},x_2)\Big|\mathbb{P}_{(\Delta T^k_{2},\eta^k_{2})}(ds^k_{2},dx_2)\\
& &\\
&\le& C\|\textbf{b}^k_1-\bar{\textbf{b}}^k_1\|_{\mathbb{R}^2} + \int_{\mathbb{S}_k}\big|\mathbb{Z}^k_{2}(\textbf{b}^k_{1},s^k_{2},x_2)- \mathbb{Z}^k_{2}(\bar{\textbf{b}}^k_{1},s^k_{2},x_2)  \big|\mathbb{P}_{(\Delta T^k_{2},\eta^k_{2})}(ds^k_{2},d\tilde{i}^k_{2}).
\end{eqnarray*}
\end{small}
Now, by using the same analysis that we did in previous steps, we shall state there exists a constant $C$ such that
\begin{small}
$$\int_{\mathbb{S}_k}\big|\mathbb{Z}^k_{2}(\textbf{b}^k_{1},s^k_{2},x_2)- \mathbb{Z}^k_{2}(\bar{\textbf{b}}^k_{1},s^k_{2},x_2)  \big|\mathbb{P}_{(\Delta T^k_{2},\eta^k_{2})}(ds^k_{2},d\tilde{i}^k_{2}) = \frac{1}{2}\int_0^{+\infty}\big|\mathbb{Z}^k_{2}(\textbf{b}^k_{1},s^k_{2},1)- \mathbb{Z}^k_{2}(\bar{\textbf{b}}^k_{1},s^k_{2},1)  \big|\theta^k(ds^k_2)
$$
$$+\frac{1}{2}\int_0^{+\infty}\big|\mathbb{Z}^k_{2}(\textbf{b}^k_{1},s^k_{2},-1)- \mathbb{Z}^k_{2}(\bar{\textbf{b}}^k_{1},s^k_{2},-1)  \big|\theta^k(ds^k_2)\le C\|\textbf{b}^k_1 - \bar{\textbf{b}}^k_1\|\quad \forall \textbf{b}^k_1,\bar{\textbf{b}}^k_1\in \widetilde{\mathbb{S}}_k.$$
\end{small}
This allows us to conclude the proof.
\end{proof}

By combining Corollary \ref{contVABS} and Theorem \ref{regPSDE}, we arrive at the following result. In the sequel,

$$K_j = \big([0,T]\times B_r(0)\big)\times \ldots\times \big([0,T]\times B_r(0)\big)\quad \text{j-fold cartesian product}$$
for $j=1, \ldots, e(k,T)-1$.

\begin{corollary}\label{comecoex1}
Let us define the sequences of approximation architectures

$$
\mathcal{H}^k_{N,j}=\Big\{ p\in \mathcal{P}_{j(d+1)}\big(N^{1/j(d+1) + 2}\big) ; \|p\|_{\infty,K_j}\le 2 \|\textbf{U}^k_j\|_{\infty,1,K_j}\Big\}, \mathcal{H}^k_{N,0} = [-L_k,L_k],
$$
for the continuation values $\mathbf{U}^k_j; 0\le j\le e(k,T)-1$ associated with the optimal stopping problem $V_0^k$ given by (\ref{valuef}) where $Z^k$ is the imbedded discrete structure (\ref{REWGAS}) associated with the reward process $F(X)$ where $X$ is the path-dependent SDE (\ref{pdsdeBM}). Suppose Assumption \textbf{III} holds true and we fix $k\ge 1$. Then, for $j=0,\ldots, e(k,T)-1$, we have

\begin{equation}\label{fest}
\mathbb{E}\|\widehat{\mathbf{U}}^k_j(\cdot; \mathbf{A}^k_{jN}) - \mathbf{U}^k_j\|_{L^2(\rho^k_j)}\le O\big(\text{log}(N)N^{\frac{-2}{2+e(k,T)-1}}\big).
\end{equation}
In particular,

$$\mathbb{E}|\widehat{V}_0(\mathbf{A}^k_{0N}) - V^k_0|\le O\big(\text{log}(N)N^{\frac{-2}{2+e(k,T)-1}}\big).
$$
\end{corollary}
\begin{proof}
From Theorem \ref{regPSDE}, we know that $\textbf{U}^k_j\in \text{W}^1\big(L^\infty(K_j)\big)$ for $j=1, \ldots, e(k,T)-1$. Then, the architecture spaces $\mathcal{H}^k_{N,j}; 0\le j\le e(k,T)-1$ they are uniformly bounded in $N$ and $j$ so that $B_k<\infty$ in assumption \textbf{(H2)}. Assumption III yields $L_k<+\infty$ for every $k\ge 1$. Then, we shall apply Corollary \ref{contVABS} to get (\ref{fest}).
\end{proof}

\begin{example}\label{EXAMPLE_PAPER}
Let $X$ be the state process given by (\ref{pdsdeBM}), where the terminal time $T=1$, the level of discretization $\epsilon_k = \varphi(k)$ for a strictly decreasing function $\varphi:[0,\infty)\rightarrow [0,\infty)$ (with inverse $\xi$) such that $\sum_{k\ge 1} \varphi^2(k)< \infty$. Let us investigate the global numerical error $\textbf{e} = \text{e}_1+ \text{e}_2$ one may occur in a non-Markovian SDE. The error $\textbf{e}$ can be decomposed as the sum of two terms: the first one ($\text{e}_1$), which was studied in \cite{LEAO_OHASHI2017.3}, is the discrete-type filtration approximation error; the second one ($\text{e}_2$), which we study in this article, it is related to the numerical approximation of the conditional expectations associated with the continuation values. We apply Corollary \ref{comecoex1} and Proposition 5.1 in \cite{LEAO_OHASHI2017.3} to state that

\begin{equation}\label{dis1}
\mathbb{E}|\widehat{V}_0(\mathbf{A}^k_{0N}) - V^k_0|= O\big(\text{log}(N)N^{\frac{-2}{2+e(k,1)-1}}\big),
\end{equation}
and
\begin{equation}\label{dis2}
|V^k_0 - S(0)| = O \Bigg(\epsilon_k^{2\beta} + \sqrt{\epsilon_k}\text{ln}~\Big(\frac{2}{\sqrt{\epsilon_k}}\Big) \Bigg),
\end{equation}
for $0 < \beta < 1$. With the estimates (\ref{dis1}) and (\ref{dis2}) at hand, we are now able to infer the amount of work (complexity) to recover the optimal value for a given level of accuracy $\textbf{e}$. Indeed, let us fix $0 < \text{e}_1 < 1$. Equation (\ref{dis2}) allows us to find the necessary number of steps related to the discretization as follows. If $0 < \beta\le 0.2$, then $\epsilon_k^{2\beta} = o\Big(\sqrt{\epsilon_k}\text{ln}~\Big(\frac{2}{\sqrt{\epsilon_k}}\Big)\Big)$. We observe $\epsilon^{2\beta}_k \le e_1 \Longleftrightarrow k \ge \xi(\text{e}^{\frac{1}{2\beta}}_1)$ and with this information at hand, we shall take $k^* = \xi(\text{e}_1^{\frac{1}{2\beta}})$. This produces

$$e(k^*,1) = \Big\lceil \frac{1}{\varphi^{2}(k^*)}\Big\rceil$$
number of steps associated with the discretization procedure. For instance, if $\varphi(k) = 2^{-k}$, $\text{e}_1 = 0.45$, $\beta=0.2$, then $\text{e}^{\frac{1}{2\beta}}_1 = 0.135$ and

$$k^*= \frac{-\text{ln}\big(\text{e}^{\frac{1}{2\beta}}_1\big)}{\text{ln}2} = 2.88.$$
This produces $e(k^*,1)= \lceil 2^{2\times 2.88}\rceil=55$ number of steps. Of course, as $\text{e}_1 \downarrow 0$, the number of steps $e(k^*,1)\uparrow +\infty$, e.g., if $\text{e}_1 = 0.3$, then $k^*=-\frac{\text{ln}~0.049}{\text{ln}~2}=4.35$, $e(k^*,1) = 416$ and so on. For a given prescribed error $0 < \text{e}_2 < 1$ and $k^{*}$, equation (\ref{dis1}) allows us to find the necessary number $N$ for the Monte-Carlo scheme in such way that $\mathbb{E}|\widehat{V}_0(\mathbf{A}^{k^*}_{0N}) - V^{k^*}_0|=O(\text{e}_2)$.
\end{example}

For an example where the state is given by a SDE driven by a fractional Brownian motion, we refer the reader to Example 5.1 in \cite{LEAO_OHASHI2017.3}.
\section{Final Remarks}
In this work, we provide a Monte Carlo algorithm for optimal stopping time problems driven by fully non-Markovian states, i.e., processes which are not reducible to vectors of Markov processes. By using statistical learning theory techniques, we present error estimates for a large class of optimal stopping time problems based on the Brownian motion filtration. Explicit bounds depend on the regularity properties of continuation
values of a suitable dynamic programming algorithm. We analyze in detail the case of reward functionals based on path-dependent SDEs and concrete linear architecture approximating spaces are also discussed.

\section{Appendix}
This section is devoted to the proof of Theorem \ref{disintegrationTH}. For simplicity of presentation, we present the proof for a bidimensional Brownian motion, i.e., we set $d=2$. We start with the following technical lemma.

\begin{lemma}\label{dislemma1}
For each non-negative integer $n$, $(j,r)\in \{1,2\}\times \{-1,1\}$, $a< b$ and $\mathbf{b}^k_n \in \mathbb{S}^n_k$, we have
$$\mathbb{P}\Big\{ \Delta T^k_{n+1}\in (a,b) \big| \aleph(\eta^k_{n+1}) = (j,r), \mathcal{A}^k_n = \mathbf{b}^k_n \Big\} = \frac{\int_a^b f_k(u+ \Delta^{k,j}_n(\mathbf{b}^k_n))du}{\int_{\Delta^{k,j}_n(\mathbf{b}^k_n)}^\infty f_k(u)du}.$$
\end{lemma}

\begin{proof}
Let us fix a non-negative integer $n$, $(j,r)\in \{1,2\}\times \{-1,1\}$, $a < b$ and $\mathbf{b}^k_n = (\textbf{i}^k_n, s^k_1, \ldots, s^k_n)$. Let us denote

$$\kappa_{j,r}(\mathbf{b}^k_n,(a,b)) =\mathbb{P}\Big\{\Delta T^k_{n+1}\in (a,b) \big| \aleph(\eta^k_{n+1}) = (j,r), \mathcal{A}^k_n = \mathbf{b}^k_n \Big\}.$$
The kernel $\kappa_{j,r}$ is a regular conditional distribution where $\Delta T^k_1, \ldots, \Delta T^k_n$ has an absolutely continuous distribution, then it is known that (for details see e.g Tjur \cite{tjur} and Prop 2.14 in Ackerman et al \cite{ackerman})

$$\kappa_{j,r}(\mathbf{b}^k_n,(a,b)) = \lim_{p\rightarrow+\infty}\mathbb{P}\Big\{\Delta T^k_{n+1}\in (a,b) \big| \aleph(\eta^k_{n+1}) = (j,r), \mathcal{A}^k_n \in\{\textbf{i}^k_n\} \times V(1/p)\Big\}$$
for every closed ball $V(1/p)$ of radius $\frac{1}{p}$ centered at $s^k_1, \ldots, s^k_n\in \mathbb{R}^n_+$. Moreover,

\begin{small}
\begin{equation}\label{inclusion}
\big\{ \aleph(\eta^k_{n+1}) = (j,r), \mathcal{A}^k_n\in \{\textbf{i}^k_n\}\times V(1/p)\big\} = \big\{ \aleph(\eta^k_{n+1}) = (j,r), \mathcal{A}^k_n\in \{\textbf{i}^k_n\}\times V(1/p), \Delta T^{k,j}_{\mathbb{j}_{j}(\textbf{i}^k_n)+1} > T^k_n  - T^{k,j}_{\mathbb{j}_{j}(\textbf{i}^k_n)}\big\}
\end{equation}
\end{small}
for every positive integer $p$. Since $(s^k_1, \ldots, s^k_n)\mapsto t^k_n(s^k_1, \ldots, s^k_n)$ is continuous, we then have

$$\kappa_{j,r}(\mathbf{b}^k_n,(a,b)) = \mathbb{P}\Big\{\Delta T^k_{n+1}\in (a,b) \big| \aleph(\eta^k_{n+1}) = (j,r), \mathcal{A}^k_n=\mathbf{b}^k_n, \Delta T^{k,j}_{\mathbb{j}_{j}(\textbf{i}^k_n)+1}> \Delta^{k,j}_n(\mathbf{b}^k_n)\Big\}.$$
Therefore,

\begin{eqnarray*}
\kappa_{j,r}(\mathbf{b}^k_n,(a,b)) &=&\mathbb{P}\Big\{\Delta T^{k,j}_{\mathbb{j}_{j}+1}\in (a+\Delta^{k,j}_n(\mathbf{b}^k_n),b+\Delta^{k,j}_n(\mathbf{b}^k_n)) \big| \aleph(\eta^k_{n+1}) = (j,r), \mathcal{A}^k_n = \mathbf{b}^k_n,\\
 & &\\
& &\Delta T^k_{\mathbb{j}_{j}(\textbf{i}^k_n)+1}> \Delta^{k,j}_n(\mathbf{b}^k_n) \Big\}\\
& &\\
&=&\mathbb{P}\Big\{\Delta T^{k,j}_{\mathbb{j}_{j}+1}\in (a+\Delta^{k,j}_n(\mathbf{b}^k_n),b+\Delta^{k,j}_n(\mathbf{b}^k_n)) \big| \mathcal{A}^k_n = \mathbf{b}^k_n, \Delta T^k_{\mathbb{j}_{j}(\textbf{i}^k_n)+1}> \Delta^{k,j}_n(\mathbf{b}^k_n) \Big\}
\end{eqnarray*}
where the last equality is due to $\mathds{1}_{\{\Delta T^{k,j}_{\mathbb{j}_{j}(\mathbf{i}^k_n)+1}\in (a+\Delta^{k,j}_n(\mathbf{b}^k_n),b+\Delta^{k,j}_n(\mathbf{b}^k_n))\}}$ is conditionally independent from $\mathds{1}_{\{\aleph(\eta^k_{n+1}) = (j,r)\}}$
given $\mathcal{A}^k_n = \mathbf{b}^k_n$. A tedious but elementary computation yields

$$
\mathbb{P}\Big\{\Delta T^{k,j}_{\mathbb{j}_{j}(\mathbf{i}^k_n)+1}\in (a+\Delta^{k,j}_n(\mathbf{b}^k_n),b+\Delta^{k,j}_n(\mathbf{b}^k_n)) \big|\mathcal{A}^k_n = \mathbf{b}^k_n,\Delta T^{k.j}_{\mathbb{j}_{j}(\mathbf{i}^k_n)+1}> \Delta^{k,j}_n(\mathbf{b}^k_n) \Big\}=
$$
$$\mathbb{P}\Big\{\Delta T^{k,j}_{\mathbb{j}_{j}(\mathbf{i}^k_n)+1}\in (a+\Delta^{k,j}_n(\mathbf{b}^k_n),b+\Delta^{k,j}_n(\mathbf{b}^k_n)) \big| \Delta T^{k,j}_{\mathbb{j}_{j}(\textbf{i}^k_n)+1}> \Delta^{k,j}_n(\mathbf{b}^k_n)\Big\}.
$$

Summing up the above steps and noticing that $\{\Delta T^{k,j}_m; m\ge 1\}$ is equally distributed with density $f_k$, we then have

\begin{equation*}
\begin{split}
&\kappa_{j,r}(\mathbf{b}^k_n, (a,b))
= \frac{\mathbb{P}\left\{\Delta T^{k,1}_1\in (a+ \Delta^{k,j}_n(\mathbf{b}^k_n) ,b+\Delta^{k,j}_n(\mathbf{b}^k_n) )\right\}}{\mathbb{P}\left\{\Delta T^{k,1}_1> \Delta^{k,j}_n(\mathbf{b}^k_n)\right\}}
=\displaystyle \frac{\displaystyle \int_{a}^{b}f_k\left(u +\Delta^{k,j}_n(\mathbf{b}^k_n) \right)du}{\displaystyle \int_{\Delta^{k,j}_n(\mathbf{b}^k_n)}^{\infty}f_k\left(u\right)du }
\end{split}
\end{equation*}
and this concludes the proof.
\end{proof}

\begin{lemma}\label{dislemma2}
For each non-negative integer $n$, $j\in \{-1,1\}$ and $\mathbf{b}^k_n\in \mathbb{S}^n_k$, we have
$$\mathbb{P}\Big\{ \aleph(\eta^k_{n+1})=(j,1)\big| \aleph_1(\eta^k_{n+1}) = j, \mathcal{A}^k_n = \mathbf{b}^k_n\Big\}=$$
$$\mathbb{P}\Big\{ \aleph(\eta^k_{n+1})=(j,-1)\big| \aleph_1(\eta^k_{n+1}) = j, \mathcal{A}^k_n = \mathbf{b}^k_n\Big\}=\frac{1}{2}.$$
\end{lemma}
\begin{proof}
Let us fix $j\in \{-1,1\}$ and $\mathbf{b}^k_n\in \mathbb{S}^n_k$. We just need to observe

$$\mathbb{P}\Big\{ \aleph(\eta^k_{n+1})=(j,1)\big| \aleph_1(\eta^k_{n+1}) = j, \mathcal{A}^k_n = \mathbf{b}^k_n\Big\}=$$
$$\mathbb{P}\Big\{ \Delta A^{k,j}(T^{k,j}_{\mathbb{j}_{j}+1})=\epsilon_k\big| \aleph_1(\eta^k_{n+1}) = j, \mathcal{A}^k_n = \mathbf{b}^k_n\Big\}=$$
$$ \mathbb{P}\Big\{ \Delta A^{k,j}(T^{k,j}_{\mathbb{j}_{j}+1})=-\epsilon_k\big| \aleph_1(\eta^k_{n+1}) = j, \mathcal{A}^k_n = \mathbf{b}^k_n\Big\}=$$
$$ \mathbb{P}\Big\{ \aleph(\eta^k_{n+1})=(j,-1)\big| \aleph_1(\eta^k_{n+1}) = j, \mathcal{A}^k_n = \mathbf{b}^k_n\Big\}=1/2.$$
\end{proof}

We observe the regular conditional probabilities are well-defined so that (see e.g Prop 2.14 in \cite{ackerman}),
$$\mathbb{P}\Big\{ \Delta T^k_{n+1}\in (a,b),\aleph(\eta^k_{n+1}) = (j,r) \big| \mathcal{A}^k_n = \mathbf{b}^k_n \Big\} = \mathbb{P}\Big\{ \Delta T^k_{n+1}\in (a,b) \big| \aleph(\eta^k_{n+1}) = (j,r), \mathcal{A}^k_n = \mathbf{b}^k_n \Big\}
$$
$$\mathbb{P}\Big\{ \aleph(\eta^k_{n+1})=(j,r)\big| \mathcal{A}^k_n = \mathbf{b}^k_n\Big\}$$
and
$$\mathbb{P}\Big\{\aleph(\eta^k_{n+1}) = (j,r) \big| \mathcal{A}^k_n = \mathbf{b}^k_n \Big\}=\mathbb{P}\Big\{\aleph_1(\eta^k_{n+1}) = j,\aleph(\eta^k_{n+1}) = (j,r) \big| \mathcal{A}^k_n = \mathbf{b}^k_n \Big\}=$$
$$\mathbb{P}\Big\{\aleph(\eta^k_{n+1}) = (j,r)  \big| \aleph_1(\eta^k_{n+1}) = j, \mathcal{A}^k_n = \mathbf{b}^k_n \Big\}\mathbb{P}\Big\{ \aleph_1(\eta^k_{n+1})=j\big| \mathcal{A}^k_n = \mathbf{b}^k_n\Big\}=$$
$$\frac{1}{2}\mathbb{P}\Big\{ \aleph_1(\eta^k_{n+1})=j\big| \mathcal{A}^k_n = \mathbf{b}^k_n\Big\},$$
where the last identity is due to Lemma \ref{dislemma2}. Therefore, Lemmas \ref{dislemma1} and \ref{dislemma2} yield

$$\mathbb{P}\Big\{ \Delta T^k_{n+1}\in (a,b),\aleph(\eta^k_{n+1}) = (j,r) \big| \mathcal{A}^k_n = \mathbf{b}^k_n \Big\}=$$
$$\frac{1}{2}\mathbb{P}\Big\{ \Delta T^k_{n+1}\in (a,b) \big| \aleph(\eta^k_{n+1}) = (j,r), \mathcal{A}^k_n = \mathbf{b}^k_n \Big\}\mathbb{P}\Big\{ \aleph_1(\eta^k_{n+1})=j\big| \mathcal{A}^k_n = \mathbf{b}^k_n\Big\}=$$
$$\frac{1}{2}\left\{\displaystyle \frac{\displaystyle \int_{a}^{b}f_k\left(u +\Delta^{k,j}_n \right)du}{\displaystyle \int_{\Delta^{k,j}_n}^{\infty}f_k\left(u\right)du }\right\}\mathbb{P}\Big\{ \aleph_1(\eta^k_{n+1})=j\big| \mathcal{A}^k_n = \mathbf{b}^k_n\Big\}$$
and to conclude the proof, we only need to compute $\mathbb{P}\Big\{ \aleph_1(\eta^k_{n+1})=j\big| \mathcal{A}^k_n = \mathbf{b}^k_n\Big\}$.

\begin{lemma}\label{dislemma3}
For each non-negative integer $n$, $j\in \{1,2\}$ and $\mathbf{b}^k_n \in \mathbb{S}^n_k$, we have
\begin{equation}\label{obv}
\mathbb{P}\Big\{ \aleph_1(\eta^k_{n+1})=j\big| \mathcal{A}^k_n = \mathbf{b}^k_n\Big\}=1
\end{equation}
for $d=1$ and

\begin{equation}\label{nobv}
\begin{split}
&\mathbb{P}\Big\{ \aleph_1(\eta^k_{n+1})=j\big| \mathcal{A}^k_n = \mathbf{b}^k_n\Big\} \\
&= \left\{\frac{\displaystyle\int_{-\infty}^{0}\int_{-s}^\infty f_{k}\big(s+t+\Delta^{k,j}_n\big)f^k_{min}(\mathbf{b}^k_n,j,t)dtds}{\displaystyle\prod_{\lambda=1}^2 \int_{\Delta^{k,\lambda}_n}^{+\infty} f_{k}(t)dt}\right\};~\text{if}~d=2.
\end{split}
\end{equation}
\end{lemma}
\begin{proof}
Let us fix $\mathbf{b}^k_n\in \mathbb{S}^n_k$ and $j\in \{1,2\}$. Identity (\ref{obv}) is obvious. For $d=2$, we observe that

\begin{eqnarray}
\nonumber\mathbb{P}\Big\{ \aleph_1(\eta^k_{n+1})=j\big| \mathcal{A}^k_n = \mathbf{b}^k_n\Big\} &=& \mathbb{P}\Big\{ T^{k,j}_{\mathbb{j}_{j}+1} < \min_{\lambda\neq j}\{T^{k,\lambda}_{\mathbb{j}_{\lambda}+1}\}\big| \mathcal{A}^k_n = \mathbf{b}^k_n\Big\}\\
\nonumber&=& \mathbb{P}\Big\{ T^{k,j}_{\mathbb{j}_{j}+1} + t^{k,j}_{\mathbb{j}_{j}} -t^{k,j}_{\mathbb{j}_{j}} - t^k_{n}< \min_{\lambda\neq j}\{ T^{k,\lambda}_{\mathbb{j}_{\lambda}+1}+ t^{k,\lambda}_{\mathbb{j}_{\lambda}} -t^{k,\lambda}_{\mathbb{j}_{\lambda}} - t^k_{n}\}\big| \mathcal{A}^k_n = \mathbf{b}^k_n\Big\}\\
\nonumber&=&\mathbb{P}\Big\{ \Delta T^{k,j}_{\mathbb{j}_{j}+1} - (t^k_{n} - t^{k,j}_{\mathbb{j}_j})< \min_{\lambda\neq j}\{\Delta T^{k,\lambda}_{\mathbb{j}_{\lambda}+1}- (t^k_n  - t^{k,\lambda}_{\mathbb{j}_{\lambda}})\}\big| \mathcal{A}^k_n = \mathbf{b}^k_n\Big\}\\
\nonumber&=&\mathbb{P}\Big\{ \Delta T^{k,j}_{\mathbb{j}_{j}+1} - (t^k_{n} - t^{k,j}_{\mathbb{j}_j})- \min_{\lambda\neq j}\{\Delta T^{k,\lambda}_{\mathbb{j}_{\lambda}+1}- (t^k_n  - t^{k,\lambda}_{\mathbb{j}_{\lambda}})\} < 0\big| \mathcal{A}^k_n = \mathbf{b}^k_n\Big\}.
\end{eqnarray}
More precisely, $\mathbb{P}\big\{ \aleph_1(\eta^k_{n+1})=j\big| \mathcal{A}^k_n = \mathbf{b}^k_n\big\}$ equals to

\begin{equation}\label{k2}
\mathbb{P}\Big\{ \Delta T^{k,j}_{\mathbb{j}_{j}(\mathbf{i}^k_n)+1} - \Delta^{k,j}_n(\mathbf{b}^k_n)- \min_{\lambda\neq j}\{\Delta T^{k,\lambda}_{\mathbb{j}_{\lambda}(\mathbf{i}^k_n)+1}- \Delta^{k,\lambda}_n(\mathbf{b}^k_n)\} < 0\big| \mathcal{A}^k_n = \mathbf{b}^k_n\Big\}.
\end{equation}

We now observe that

$$\Big\{\mathcal{A}^k_n \in \{\textbf{i}^k_n\}\times V(1/p) \Big\} = \Big\{\mathcal{A}^k_n \in \{\textbf{i}^k_n\}\times V(1/p) \Big\}\cap \big\{\Delta T^{k,j}_{\mathbb{j}_{j}(\mathbf{i}^k_n)+1} \ge  T^k_n -  T^{k,j}_{\mathbb{j}_{j}(\mathbf{i}^k_n)}~\forall j=1,2\}$$
for every closed ball $V(1/p)$ of radius $1/p$ centered at $s^k_1, \ldots, s^k_n\in \mathbb{R}^n_+$. Therefore, Prop 2.14 in \cite{ackerman} allows us to state that (\ref{k2}) equals to
\begin{equation}\label{k3}
\mathbb{P}\Big\{ \Delta T^{k,j}_{\mathbb{j}_{j}(\mathbf{i}^k_n)+1} - \Delta^{k,j}_n(\mathbf{b}^k_n)- \min_{\lambda\neq j}\{\Delta T^{k,\lambda}_{\mathbb{j}_{\lambda}(\mathbf{i}^k_n)+1}- \Delta^{k,\lambda}_n(\mathbf{b}^k_n)\} < 0\big| \mathcal{A}^k_n = \mathbf{b}^k_n, \cap_{\lambda=1}^2 \{\Delta T^{k,\lambda}_{\mathbb{j}_{\lambda}(\mathbf{i}^k_n)+1}\ge \Delta^{k,\lambda}_n(\mathbf{b}^k_n)\}\Big\}.
\end{equation}
We shall argue similarly to the proof of Lemma \ref{dislemma1} to state that (\ref{k3}) equals to

$$\bar{\kappa}(\mathbf{b}^k_n) := \mathbb{P}\Big\{ \Delta T^{k,j}_{\mathbb{j}_{j}(\mathbf{i}^k_n)+1} - \Delta^{k,j}_n(\mathbf{b}^k_n)- \min_{\lambda\neq j}\{\Delta T^{k,\lambda}_{\mathbb{j}_{\lambda}(\mathbf{i}^k_n)+1}- \Delta^{k,\lambda}_n(\mathbf{b}^k_n)\} < 0\big|\cap_{\lambda=1}^2 \{\Delta T^{k,\lambda}_{\mathbb{j}_{\lambda}(\mathbf{i}^k_n)+1}\ge \Delta^{k,\lambda}_n(\mathbf{b}^k_n)\}\Big\}.$$
This shows that $\bar{\kappa}(\mathcal{A}^k_n)$ is a version of the conditional expectation $\mathbb{P}\big\{ \aleph_1(\eta^k_{n+1})=j\big| \mathcal{A}^k_n\big\}$. In particular, if we denote $\bar{\eta}^k_n = (\eta^k_1, \ldots, \eta^k_n)$ and noticing that

$$\Big\{ \Delta T^{k,j}_{\mathbb{j}_{j}(\bar{\eta}^k_n)+1} - \Delta^{k,j}_n(\mathcal{A}^k_n)- \min_{\lambda\neq j}\{\Delta T^{k,\lambda}_{\mathbb{j}_{\lambda}(\bar{\eta}^k_n)+1}- \Delta^{k,\lambda}_n(\mathcal{A}^k_n)\} < 0\Big\}\subset \cap_{\lambda=1}^2 \{\Delta T^{k,\lambda}_{\mathbb{j}_{\lambda}(\bar{\eta}^k_n)+1}\ge \Delta^{k,\lambda}_n(\mathcal{A}^k_n)\}$$
almost surely, we can actually choose a version as


$$\bar{\kappa}(\mathbf{b}^k_n) = \frac{\mathbb{P}\Big\{ \Delta T^{k,j}_{\mathbb{j}_{j}(\mathbf{i}^k_n)+1} - \Delta^{k,j}_n(\mathbf{b}^k_n)- \min_{\lambda\neq j}\{\Delta T^{k,\lambda}_{\mathbb{j}_{\lambda}(\mathbf{i}^k_n)+1}- \Delta^{k,\lambda}_n(\mathbf{b}^k_n)\} < 0\Big\}}{\mathbb{P}\Big\{\cap_{\lambda=1}^2 \{\Delta T^{k,\lambda}_{\mathbb{j}_{\lambda}(\mathbf{i}^k_n)+1}\ge \Delta^{k,\lambda}_n(\mathbf{b}^k_n)\}\Big  \}}; \mathbf{b}^k_n \in \mathbb{S}^n_k.$$
Lastly, it remains to compute $\mathbb{P}\big\{ \Delta T^{k,j}_{\mathbb{j}_{j}(\mathbf{i}^k_n)+1} - \Delta^{k,j}_n(\mathbf{b}^k_n)- \min_{\lambda\neq j}\{\Delta T^{k,\lambda}_{\mathbb{j}_{\lambda}(\mathbf{i}^k_n)+1}- \Delta^{k,\lambda}_n(\mathbf{b}^k_n)\} < 0\big\}$, but this is a straightforward application of the jacobian method and the fact that $\Delta T^{k,j}_{\mathbb{j}_{j}(\mathbf{i}^k_n)+1}$ is independent from $\min_{\lambda\neq j}\{\Delta T^{k,\lambda}_{\mathbb{j}_{\lambda}(\mathbf{i}^k_n)+1}- \Delta^{k,\lambda}_n(\mathbf{b}^k_n)\}$ and both random variables are absolutely continuous for every $\mathbf{b}^k_n$. We conclude the proof.
\end{proof}
From Lemmas \ref{dislemma1}, \ref{dislemma2} and \ref{dislemma3}, we then conclude the proof of Theorem \ref{disintegrationTH}.

\

\noindent \textbf{Acknowledgments.} Alberto Ohashi acknowledges the support of CNPq-Bolsa de Produtividade de Pesquisa grant 303443/2018-9. Francesco Russo and Alberto Ohashi acknowledge the financial support of Math Amsud grant 88887.197425/2018-00. Francys de Souza acknowledges the support of S\~ao Paulo Research Foundation (FAPESP) grant 2017/23003-6.

\end{document}